\documentclass[12pt,american]{article}
\usepackage[T1]{fontenc}
\usepackage[utf8]{inputenc}
\usepackage{color}
\usepackage{babel}
\usepackage{mathtools}
\usepackage{amsmath}
\usepackage{amsthm}
\usepackage{amssymb}
\usepackage[pdfusetitle,
 bookmarks=true,bookmarksnumbered=false,bookmarksopen=false,
 breaklinks=false,pdfborder={0 0 0},pdfborderstyle={},backref=false,colorlinks=true]
 {hyperref}
\hypersetup{
 linkcolor=magenta, urlcolor=blue, citecolor=blue}

\makeatletter
\theoremstyle{plain}
\newtheorem{thm}{\protect\theoremname}
\theoremstyle{plain}
\newtheorem{lyxalgorithm}[thm]{\protect\algorithmname}
\theoremstyle{plain}
\newtheorem{lem}[thm]{\protect\lemmaname}
\theoremstyle{remark}
\newtheorem{rem}[thm]{\protect\remarkname}
\theoremstyle{plain}
\newtheorem{cor}[thm]{\protect\corollaryname}
\theoremstyle{plain}
\newtheorem{prop}[thm]{\protect\propositionname}

\DeclareMathOperator{\Tr}{Tr}

\DeclareMathOperator{\sgn}{sgn}

\allowdisplaybreaks

\usepackage{fullpage}

\numberwithin{equation}{section}

\makeatother

\providecommand{\algorithmname}{Algorithm}
\providecommand{\corollaryname}{Corollary}
\providecommand{\lemmaname}{Lemma}
\providecommand{\propositionname}{Proposition}
\providecommand{\remarkname}{Remark}
\providecommand{\theoremname}{Theorem}

\begin{document}
\title{Accelerated optimization of measured relative entropies}
\author{Zixin Huang\thanks{School of Science, College of STEM, RMIT University, Melbourne, Australia}
\and Mark M. Wilde\thanks{School of Electrical and Computer Engineering,
Cornell University, Ithaca, New York 14850, USA}}
\date{\today}
\maketitle
\begin{abstract}
The measured relative entropy and measured R\'enyi relative entropy
are quantifiers of the distinguishability of two quantum states $\rho$
and $\sigma$. They are defined as the maximum classical relative
entropy or R\'enyi relative entropy realizable by performing a measurement
on $\rho$ and $\sigma$, and they have interpretations in terms of
asymptotic quantum hypothesis testing. Crucially, they can be rewritten
in terms of variational formulas involving the optimization of a concave
or convex objective function over the set of positive definite operators.
In this paper, we establish foundational properties of these objective
functions by analyzing their matrix gradients and Hessian superoperators;
namely, we prove that these objective functions are $\beta$-smooth
and $\gamma$-strongly convex / concave, where $\beta$ and $\gamma$
depend on the max-relative entropies of $\rho$ and $\sigma$. A practical
consequence of these properties is that we can conduct Nesterov accelerated
projected gradient descent / ascent, a well known classical optimization
technique, to calculate the measured relative entropy and measured
R\'enyi relative entropy to arbitrary precision. These algorithms
are generally more memory efficient than our previous algorithms based
on semidefinite optimization {[}Huang and Wilde, arXiv:2406.19060{]},
and for well conditioned states $\rho$ and $\sigma$, these algorithms
are notably faster.
\end{abstract}
\tableofcontents{}

\section{Introduction}

\subsection{Background and motivation}

The measured relative entropy is a quantifier of the distinguishability
of two quantum states~\cite{Donald1986,Piani2009}, defined as the
maximum of the classical relative entropy of the probability distributions
resulting from performing the same measurement on the states (see
\eqref{eq:meas-rel-ent-def} for a definition). This quantity has
appeared in a number of contexts in quantum information theory~\cite{petz2007quantum,Piani2009,Mosonyi2015,Matsumoto2016,Hiai2017,Berta2017,Hiai2021,Huang2024}
and has an operational meaning in asymptotic quantum hypothesis testing.
There has been interest in developing techniques for calculating this
quantity, formulated as classical~\cite{Huang2024} and quantum algorithms
\cite{Goldfeld2024}, and one requires methods from optimization theory
for doing so. Critical to this optimization task is a variational
formula for measured relative entropy~\cite{Berta2017}, which rewrites
the quantity in terms of an optimization of a concave objective function
over the cone of positive definite operators.

Similar to the above, the measured R\'enyi relative entropy is also
a distinguishability quantifier \cite[Eqs.~(3.116)--(3.117)]{Fuchs1996},
defined similarly as stated above but in terms of the R\'enyi relative
entropy (see \eqref{eq:measured-renyi-rel-ent-def} for a definition).
There has also been interest in developing techniques for calculating
this quantity~\cite{Goldfeld2024,Huang2024}, again requiring methods
from optimization theory. Likewise, there are variational formulas
for measured R\'enyi relative entropy~\cite{Berta2017}, holding
for all values of the R\'enyi parameter and expressing it as an optimization
of a convex or concave objective function over the cone of positive
definite operators (in particular, the objective function is convex
if the R\'enyi parameter $\alpha\in\left(0,1\right)$ and it is concave
if $\alpha\in\left(1,\infty\right)$).

In our prior work~\cite{Huang2024}, we formulated the optimization
task of the measured relative entropy and the measured R\'enyi relative
entropy in terms of semidefinite optimization. An advantage of this
approach is that standard semidefinite programming (SDP) solvers
\cite{sedumi1999,cvx2014,cvxpy2018,mosek2023,yalmip2022} can be used
for performing the needed optimization. However, a drawback of this
approach for calculating measured relative entropy is that, being
based on the methods of~\cite{Fawzi2019}, the resulting SDPs require
a number of matrix inequalities that is logarithmic in the inverse
precision desired. Furthermore, when calculating measured R\'enyi
relative entropy, the approach relies on techniques from~\cite{Fawzi2017},
and the resulting SDPs require a number of matrix inequalities that
is logarithmic in the denominator resulting from a rational approximation
of the R\'enyi parameter, and thus the number of matrix inequalities
becomes larger as the R\'enyi parameter approaches one. Additionally,
in both approaches, the dimension of every matrix involved in each
matrix inequality is twice that of the original states. In practice,
both approaches work well for moderately sized states, but it has
been left open to determine if there could be improvements.

\subsection{Summary of contributions}

In order to address the aforementioned concerns, here we study foundational
properties of the objective functions involved in the variational
formulas for the measured relative entropy and the measured R\'enyi
relative entropy. Namely, we establish analytical expressions for
the matrix gradients and Hessian superoperators of these objective
functions (see Lemmas~\ref{lem:matrix-gradient-meas-rel-ent},~\ref{lem:Hessian-superoperator-measured-rel-ent},
\ref{lem:matrix-gradient-meas-renyi-rel-ent}, and~\ref{lem:hessian-super-op-meas-renyi-rel-ent}),
and we then prove that they are $\beta$-smooth and $\gamma$-strongly
convex / concave on an operator interval containing the optimal positive
definite operator (see Corollaries~\ref{cor:smoothness-strong-concavity-meas-rel-ent}
and~\ref{cor:smoothness-strong-concavity-meas-renyi-rel-ent}).

As a consequence of these findings, we can conduct Nesterov accelerated
projected gradient descent / ascent, a well known classical optimization
technique~\cite{Nesterov1983,Nesterov2004,Beck2009}, to calculate
the measured relative entropy and measured R\'enyi relative entropy
to arbitrary precision (see Algorithms~\ref{alg:measured-rel-ent},
\ref{alg:measured-renyi-0-half}, and~\ref{alg:measured-renyi-half-infty}).
The resulting algorithms are generally more memory efficient than
our previous SDP algorithm from~\cite{Huang2024}, and for well-conditioned
states, they are notably faster. In particular, each matrix gradient
update step of the algorithms here involves matrices that have the
same dimension as that of the states, and there is no need to perform
the costly linear system solver used in each iteration of the interior-point
method for semidefinite optimization~\cite{Nesterov1994} (see also
\cite{Vandenberghe1996,Helmberg1996,Todd2001,Helmberg2003}). Based
on the general fact that Nesterov accelerated gradient descent is
known to be an optimal method for minimizing $\beta$-smooth and $\gamma$-strongly
convex functions \cite[Section~3.7 and Theorem~3.15]{Bubeck2015},
we suspect that our optimization method proposed here is likely to
be optimal (i.e., fastest possible) for well-conditioned states.

\subsection{Paper organization}

Our paper is structured as follows. In Section~\ref{sec:Preliminaries},
we review background material needed to understand the rest of the
paper. In particular, we review the concepts of matrix gradient and
Hessian superoperator of a smooth function $f$ of a positive definite
matrix $\omega$. We include some basic examples, then we review the
notions of $\beta$-smoothness and $\gamma$-strong convexity / concavity
for such a function, and we finally review Nesterov accelerated projected
gradient descent / ascent for optimizing the function $f$. Section~\ref{sec:Measured-relative-entropy}
presents our results for measured relative entropy: namely, for the
objective function corresponding to the variational formula in \eqref{eq:measured-rel-ent-var-form},
we derive its matrix gradient, an operator interval containing the
optimal $\omega$, its Hessian superoperator, its $\beta$-smoothness
and $\gamma$-strong concavity on the aforementioned operator interval,
and Nesterov accelerated projected gradient descent for performing
the optimization. Section~\ref{sec:Measured-Renyi-relative} presents
similar results for the measured R\'enyi relative entropy, using
the variational formula in \eqref{eq:meas-renyi-var-form} and breaking
up the development based on three different intervals for the R\'enyi
parameter $\alpha$: including $\alpha\in\left(0,\frac{1}{2}\right)$,
$\alpha\in\left[\frac{1}{2},1\right)$, and $\alpha\in\left(1,\infty\right)$,
with it often being possible to merge developments for the latter
two intervals due to common aspects. Section~\ref{sec:Comparison-with-SDP}
provides a detailed comparison of our previous algorithms from~\cite{Huang2024}
with our approach, with the main conclusion being that our approach
here is notably faster than our previous approach whenever the states
are well conditioned. We finally conclude in Section~\ref{sec:Conclusion}
with a summary and suggestions for future research. Appendix~\ref{app:matrix-derivatives}
contains review material on matrix derivatives and important examples
relevant for our purposes here. Appendix~\ref{app:Derivations-of-integral-formulas}
contains derivations of integral formulas used for analyzing the $\beta$-smoothness
and $\gamma$-strong convexity / concavity of the objective functions
relevant for measured R\'enyi relative entropy.

\section{Preliminaries}

\label{sec:Preliminaries}As mentioned above, one main goal of our
paper is to establish algorithms for optimizing the measured relative
entropy and the measured R\'enyi relative entropy. Due to results
of~\cite{Berta2017}, these quantities can be written as an optimization
of either a convex or concave function of a positive definite matrix.
Our results here strengthen these properties, demonstrating smoothness
and strong convexity / concavity of these functions.

With this in mind, we begin by introducing a general context for optimizing
a scalar, real-valued function of a positive definite matrix, which
applies to all of the optimization tasks considered in our paper.
In the theory of optimization, an ideal scenario occurs when the function
of interest has certain smoothness and convexity / concavity properties,
as reviewed below.

\subsection{Matrix gradient}

Let $f(\omega)$ be a scalar, real-valued function of a positive definite
operator $\omega$. The matrix gradient or matrix derivative of $f$
with respect to $\omega$ is defined in terms of the partial derivative
of $f$ with respect to each matrix element $\omega_{i,j}$ of $\omega$,
denoted as $\frac{\partial f(\omega)}{\partial\omega_{i,j}}$. For
$d\in\mathbb{N}$ and a $d\times d$ matrix $\omega$, these terms
can be arranged into a matrix, using a convenient ordering called
numerator layout notation \cite[Chapter~9]{Magnus2019}, as follows:
\begin{equation}
\frac{\partial f(\omega)}{\partial\omega}\coloneqq\begin{bmatrix}\frac{\partial f(\omega)}{\partial\omega_{1,1}} & \frac{\partial f(\omega)}{\partial\omega_{2,1}} & \cdots & \frac{\partial f(\omega)}{\partial\omega_{d,1}}\\
\frac{\partial f(\omega)}{\partial\omega_{1,2}} & \frac{\partial f(\omega)}{\partial\omega_{2,2}} &  & \frac{\partial f(\omega)}{\partial\omega_{d,2}}\\
\vdots &  & \ddots & \vdots\\
\frac{\partial f(\omega)}{\partial\omega_{1,d}} & \frac{\partial f(\omega)}{\partial\omega_{2,d}} & \cdots & \frac{\partial f(\omega)}{\partial\omega_{d,d}}
\end{bmatrix}.\label{eq:matrix-gradient-def}
\end{equation}
Indeed, the elements of the matrix derivative of $f(\omega)$ form
a matrix $\frac{\partial f(\omega)}{\partial\omega}$ with matrix
elements written in the following alternative ways:
\begin{equation}
\left[\frac{\partial f(\omega)}{\partial\omega}\right]_{i,j}=\langle j|\frac{\partial f(\omega)}{\partial\omega}|i\rangle=\frac{\partial f(\omega)}{\partial\omega_{j,i}},
\end{equation}
where $\left\{ |i\rangle\right\} _{i}$ is the standard (orthonormal)
basis. See Appendix~\ref{app:matrix-derivatives} for a review of
various aspects of matrix derivatives, along with formulas for matrix
derivatives of example functions that we use throughout the paper.
As a simple example, consider that
\begin{align}
\frac{\partial}{\partial\omega_{i,j}}\omega & =\frac{\partial}{\partial\omega_{i,j}}\sum_{k,\ell}\omega_{k,\ell}|k\rangle\!\langle\ell|\label{eq:basic-matrix-deriv-1}\\
 & =\sum_{k,\ell}\frac{\partial}{\partial\omega_{i,j}}\omega_{k,\ell}|k\rangle\!\langle\ell|\\
 & =\sum_{k,\ell}\delta_{i,k}\delta_{j,\ell}|k\rangle\!\langle\ell|\\
 & =|i\rangle\!\langle j|,\label{eq:basic-matrix-deriv-last}
\end{align}
which implies that
\begin{align}
\frac{\partial}{\partial\omega_{j,i}}\Tr\!\left[A\omega\right] & =\Tr\!\left[A\frac{\partial}{\partial\omega_{i,j}}\omega\right]\\
 & =\Tr\!\left[A|j\rangle\!\langle i|\right]\\
 & =\langle i|A|j\rangle,
\end{align}
so that $\frac{\partial}{\partial\omega}\Tr\!\left[A\omega\right]=A$.
As another example, consider that
\begin{align}
\frac{\partial}{\partial\omega_{j,i}}\Tr\!\left[A\omega^{2}\right] & =\frac{\partial}{\partial\omega_{j,i}}\Tr\!\left[A\omega\omega\right]\\
 & =\Tr\!\left[A\left(\frac{\partial}{\partial\omega_{j,i}}\omega\right)\omega\right]+\Tr\!\left[A\omega\left(\frac{\partial}{\partial\omega_{j,i}}\omega\right)\right]\\
 & =\Tr\!\left[\omega A\left(\frac{\partial}{\partial\omega_{j,i}}\omega\right)\right]+\Tr\!\left[A\omega\left(\frac{\partial}{\partial\omega_{j,i}}\omega\right)\right]\\
 & =\Tr\!\left[\left(\omega A+A\omega\right)\left(\frac{\partial}{\partial\omega_{j,i}}\omega\right)\right]\\
 & =\Tr\!\left[\left(\omega A+A\omega\right)|j\rangle\!\langle i|\right]\\
 & =\langle i|\left(\omega A+A\omega\right)|j\rangle,
\end{align}
which implies that
\begin{equation}
\frac{\partial}{\partial\omega}\Tr\!\left[A\omega^{2}\right]=\omega A+A\omega.
\end{equation}
If $A$ is Hermitian, then the functions $\Tr\!\left[A\omega\right]$
and $\Tr\!\left[A\omega^{2}\right]$ are real-valued for positive
definite $\omega$, consistent with our assumption on $f$.

\subsection{Hessian superoperator}

The second matrix derivative of $f(\omega)$ has elements defined
as follows:
\begin{equation}
\frac{\partial}{\partial\omega_{k,\ell}}\frac{\partial}{\partial\omega_{i,j}}f(\omega).
\end{equation}
Generalizing how the matrix derivative is a rank-two tensor, i.e.,
a matrix, the second derivative is a rank-four tensor (with four indices),
which is in correspondence with a superoperator, called the Hessian
superoperator $\mathcal{H}_{\omega}$. We index the elements of the
Hessian superoperator~$\mathcal{H}_{\omega}$ as follows:
\begin{equation}
\langle j|\mathcal{H}_{\omega}(|k\rangle\!\langle\ell|)|i\rangle=\frac{\partial}{\partial\omega_{k,\ell}}\frac{\partial}{\partial\omega_{i,j}}f(\omega),
\end{equation}
where $|k\rangle$ and $|\ell\rangle$ are chosen from the standard
orthonormal basis $\left\{ |i\rangle\right\} _{i}$. Due to the fact
that
\begin{equation}
\frac{\partial}{\partial\omega_{k,\ell}}\frac{\partial}{\partial\omega_{i,j}}f(\omega)=\frac{\partial}{\partial\omega_{i,j}}\frac{\partial}{\partial\omega_{k,\ell}}f(\omega),
\end{equation}
for all $i,j,k,\ell$, the Hessian superoperator $\mathcal{H}_{\omega}$
is self-adjoint, in the following sense:
\begin{equation}
\langle j|\mathcal{H}_{\omega}(|k\rangle\!\langle\ell|)|i\rangle=\langle\ell|\mathcal{H}_{\omega}(|i\rangle\!\langle j|)|k\rangle.
\end{equation}
This is equivalent to
\begin{equation}
\Tr\!\left[|i\rangle\!\langle j|\mathcal{H}_{\omega}(|k\rangle\!\langle\ell|)\right]=\Tr\!\left[|k\rangle\!\langle\ell|\mathcal{H}_{\omega}\!\left(|i\rangle\!\langle j|\right)\right].
\end{equation}
Continuing the second example above, we can determine its Hessian
superoperator by observing that
\begin{align}
\frac{\partial}{\partial\omega_{k,\ell}}\frac{\partial}{\partial\omega}\Tr\!\left[A\omega^{2}\right] & =\frac{\partial}{\partial\omega_{k,\ell}}\left(\omega A+A\omega\right)\\
 & =\left(\frac{\partial}{\partial\omega_{k,\ell}}\omega\right)A+A\left(\frac{\partial}{\partial\omega_{k,\ell}}\omega\right)\\
 & =|k\rangle\!\langle\ell|A+A|k\rangle\!\langle\ell|.
\end{align}
Extending the last line by linearity, this implies that the Hessian
superoperator $\mathcal{H}_{\omega}$ for this example is specified
by its action on an arbitrary matrix $X$ as follows:
\begin{equation}
\mathcal{H}_{\omega}(X)=XA+AX.
\end{equation}

Using the matrix gradient and the Hessian superoperator, the function
$f(\omega)$ has a Taylor series expansion as follows:
\begin{equation}
f(\omega+\tau)=f(\omega)+\left\langle \tau,\frac{\partial f(\omega)}{\partial\omega}\right\rangle +\frac{1}{2}\left\langle \tau,\mathcal{H}_{\omega}(\tau)\right\rangle +o\!\left(\left\Vert \tau\right\Vert _{2}^{2}\right),\label{eq:Taylor-expansion-function-of-matrix}
\end{equation}
where the Hilbert--Schmidt inner product and norm are respectively
defined for $X,Y$ as
\begin{align}
\left\langle X,Y\right\rangle  & \coloneqq\Tr\!\left[X^{\dag}Y\right],\\
\left\Vert X\right\Vert _{2} & \coloneqq\sqrt{\left\langle X,X\right\rangle }.
\end{align}

We should finally indicate that there is a direct correspondence between
these notions and the usual notions of gradient and Hessian matrix.
Namely, we can stack the columns of \eqref{eq:matrix-gradient-def}
(usually called ``vec'' operation), and the matrix gradient becomes
the conventional gradient vector. Furthermore, the Hessian superoperator
can be reorganized as a Hessian matrix using what is called the natural
representation of superoperators \cite[Section~2.2.2]{Watrous2018},
and then there is no distinction between the above notions and the
conventional ones. In more detail, by defining the maximally entangled
vector $|\Gamma\rangle$ as
\begin{equation}
|\Gamma\rangle\coloneqq\sum_{i}|i\rangle\otimes|i\rangle,
\end{equation}
the matrix gradient $\frac{\partial f(\omega)}{\partial\omega}$ corresponds
to the vector gradient $\left(\frac{\partial f(\omega)}{\partial\omega}\otimes I\right)|\Gamma\rangle$,
and there exists a Hessian matrix $H_{\omega}$ satisfying
\begin{equation}
\left\langle \tau,\mathcal{H}_{\omega}(\tau)\right\rangle =\langle\tau|H_{\omega}|\tau\rangle,
\end{equation}
for all $\tau$, where $|\tau\rangle\coloneqq\left(\tau\otimes I\right)|\Gamma\rangle$.
Furthermore, under this mapping, the Hilbert--Schmidt inner product
reduces to the usual Euclidean inner product and the Hilbert--Schmidt
norm reduces to the usual Euclidean norm. The main advantage of working
directly with the matrix gradient and the Hessian superoperator is
that we can more easily derive smoothness and concavity properties
of $f(\omega)$.

\subsection{Smoothness and strong convexity / concavity}

Suppose that $f$ is a convex and smooth function of $\omega$, so
that the expansion in \eqref{eq:Taylor-expansion-function-of-matrix}
exists. Directly related to definitions for smooth functions of vectors,
the function $f$ is called $\beta$-smooth if the largest eigenvalue
of the Hessian superoperator is no larger than $\beta>0$ \cite[Lemma~2.26]{Garrigos2024};
i.e., if
\begin{equation}
\max_{Y:\left\Vert Y\right\Vert _{2}=1}\left\langle Y,\mathcal{H}_{\omega}(Y)\right\rangle \leq\beta,
\end{equation}
and it is called $\gamma$-strongly convex if the smallest eigenvalue
of the Hessian superoperator is no smaller than $\gamma>0$ \cite[Lemma~2.15]{Garrigos2024};
i.e., if
\begin{equation}
\min_{Y:\left\Vert Y\right\Vert _{2}=1}\left\langle Y,\mathcal{H}_{\omega}(Y)\right\rangle \geq\gamma.
\end{equation}
In the above, we have defined the largest and smallest eigenvalues
of the superoperator $\mathcal{H}_{\omega}$ in a variational way
as one can do for matrices, but instead replacing the standard vector
Euclidean inner product with the Hilbert--Schmidt inner product.
In this sense, all of the eigenvalues of the Hessian superoperator
of a $\beta$-smooth and $\gamma$-strongly convex function $f$ are
strictly positive and contained in the interval $\left[\gamma,\beta\right]$.

Complementary to the above, suppose that $f$ is a concave and smooth
function of $\omega$, so that the expansion in \eqref{eq:Taylor-expansion-function-of-matrix}
exists. The function $f$ is called $\beta$-smooth if the smallest
eigenvalue of the Hessian superoperator is no smaller than $-\beta<0$;
i.e., if
\begin{equation}
\min_{Y:\left\Vert Y\right\Vert _{2}=1}\left\langle Y,\mathcal{H}_{\omega}(Y)\right\rangle \geq-\beta,
\end{equation}
and it is called $\gamma$-strongly concave if the largest eigenvalue
of the Hessian superoperator is no larger than $-\gamma<0$; i.e.,
if
\begin{equation}
\max_{Y:\left\Vert Y\right\Vert _{2}=1}\left\langle Y,\mathcal{H}_{\omega}(Y)\right\rangle \leq-\gamma.
\end{equation}
All of the eigenvalues of the Hessian superoperator of a $\beta$-smooth
and $\gamma$-strongly concave function $f$ are strictly negative
and contained in the interval $\left[-\beta,-\gamma\right]$.

It can also be the case that a function $f$ is $\beta$-smooth and
$\gamma$-strongly convex on an operator interval. In fact, this occurs
in all of our developments that follow. That is, let $\mathcal{I}\subset\mathbb{R}_{+}$
be a closed interval. Then it can be the case $f$ is $\beta$-smooth
and $\gamma$-strongly convex on $\omega$ having all of its eigenvalues
in the interval $\mathcal{I}$.

\subsection{Nesterov accelerated optimization of smooth and strongly convex /
concave functions}

A function $f(\omega)$ of a positive definite matrix $\omega$ that
is both $\beta$-smooth and $\gamma$-strongly convex on an operator
interval is ideal for optimization. That is, there exists a well known
procedure called Nesterov accelerated projected gradient descent to
optimize such a function~\cite{Nesterov2004,Beck2009}. Let
\begin{equation}
\kappa\coloneqq\frac{\beta}{\gamma}
\end{equation}
denote the condition number of $f$. Let $A$ be a Hermitian matrix
with spectral decomposition $A=\sum_{i}\lambda_{i}\Pi_{i}$. We define
the clamping of $A$ to a closed interval $\mathcal{I}\subset\mathbb{R}$
as follows:
\begin{equation}
\left[A\right]_{\mathcal{I}}\coloneqq\sum_{i:\lambda_{i}\in\mathcal{I}}\lambda_{i}\Pi_{i}.
\end{equation}
Suppose that $f(\omega)$ is $\beta$-smooth and $\gamma$-strongly
convex for every $\omega$ having all of its eigenvalues contained
in the interval $\mathcal{I}\subset\mathbb{R}_{+}$, and suppose furthermore
that the optimal $\omega^{\star}$ has all of its eigenvalues contained
in the interval $\mathcal{I}$.

The Nesterov accelerated projected gradient descent algorithm for
minimizing $f$ to within $\varepsilon>0$ error of the optimal solution
is as follows:
\begin{lyxalgorithm}
\label{alg:NAPGD}Proceed according to the following steps:
\begin{enumerate}
\item Set $m=0$, and for some $\delta\in\mathcal{I}$, initialize
\begin{equation}
\omega_{0}\leftarrow\xi_{0}\leftarrow\delta I.
\end{equation}
\item While $\left\Vert \left.\frac{\partial}{\partial\omega}f(\omega)\right|_{\omega=\omega_{m}}\right\Vert _{2}>\sqrt{2\gamma\varepsilon}$,
\begin{align}
\xi_{m+1} & \leftarrow\omega_{m}-\frac{1}{\beta}\left.\frac{\partial}{\partial\omega}f(\omega)\right|_{\omega=\omega_{m}},\\
\chi_{m+1} & \leftarrow\left[\xi_{m+1}\right]_{\mathcal{I}},\\
\omega_{m+1} & \leftarrow\chi_{m+1}+\left(\frac{\sqrt{\kappa}-1}{\sqrt{\kappa}+1}\right)\left(\chi_{m+1}-\chi_{m}\right).
\end{align}
Set $m=m+1$.
\end{enumerate}
\end{lyxalgorithm}

The stopping condition $\left\Vert \left.\frac{\partial}{\partial\omega}f(\omega)\right|_{\omega=\omega_{m}}\right\Vert _{2}\leq\sqrt{2\gamma\varepsilon}$
implies that 
\begin{equation}
\left|f(\omega^{\star})-f(\omega_{m})\right|\leq\varepsilon,\label{eq:error-threshold-desired}
\end{equation}
as a consequence of strong convexity. Indeed, an equivalent definition
of strong convexity for a smooth function $f$ is as follows \cite[Lemma~2.14]{Garrigos2024}:
\begin{equation}
f(\omega_{y})\geq f(\omega_{x})+\left\langle \left.\frac{\partial}{\partial\omega}f(\omega)\right|_{\omega=\omega_{x}},\omega_{y}-\omega_{x}\right\rangle +\frac{\gamma}{2}\left\Vert \omega_{y}-\omega_{x}\right\Vert _{2}^{2}.\label{eq:strong-convexity-def}
\end{equation}
Now minimizing both sides of \eqref{eq:strong-convexity-def} with
respect to $\omega_{y}$, and noting that the right-hand side attains
its minimum at $\omega_{y}=\omega_{x}-\frac{1}{\gamma}\left.\frac{\partial}{\partial\omega}f(\omega)\right|_{\omega=\omega_{x}}$
(while enlarging the minimization on the right-hand side to include
Hermitian matrices), leads to the inequality
\begin{equation}
f(\omega^{\star})\geq f(\omega_{x})-\frac{1}{2\gamma}\left\Vert \left.\frac{\partial}{\partial\omega}f(\omega)\right|_{\omega=\omega_{x}}\right\Vert _{2}^{2},
\end{equation}
which implies \eqref{eq:error-threshold-desired} when rearranged.

The number of iterations required by Algorithm~\ref{alg:NAPGD} is
equal to \cite[Theorem~3.18]{Bubeck2015}
\begin{equation}
O\!\left(\sqrt{\kappa}\ln\!\left(\frac{1}{\varepsilon}\right)\right).
\end{equation}

A similar procedure can be employed with the same guarantees if the
function $f$ is instead $\beta$-smooth and $\gamma$-strongly concave:
the only change is that the gradient update step should be $\xi_{m+1}\leftarrow\omega_{m}+\frac{1}{\beta}\left.\frac{\partial}{\partial\omega}f(\omega)\right|_{\omega=\omega_{m}}$,
so that it corresponds to an ascent step searching for a maximum,
rather than a descent step searching for a minimum.

\section{Measured relative entropy}

\label{sec:Measured-relative-entropy}For a positive definite state
$\rho$ and a positive definite operator $\sigma$, the measured relative
entropy is defined as~\cite{Donald1986,Piani2009}
\begin{equation}
D^{M}(\rho\|\sigma)\coloneqq\sup_{\left(\Lambda_{x}\right)_{x\in\mathcal{X}}}\sum_{x\in\mathcal{X}}\Tr\!\left[\Lambda_{x}\rho\right]\ln\!\left(\frac{\Tr\!\left[\Lambda_{x}\rho\right]}{\Tr\!\left[\Lambda_{x}\sigma\right]}\right),\label{eq:meas-rel-ent-def}
\end{equation}
where $\left(\Lambda_{x}\right)_{x\in\mathcal{X}}$ is a positive
operator-valued measure (i.e., satisfying $\Lambda_{x}\geq0$ for
all $x\in\mathcal{X}$ and $\sum_{x\in\mathcal{X}}\Lambda_{x}=I$).
It can be expressed in the following variational form \cite[Lemma~1 and Theorem~2]{Berta2017}:
\begin{equation}
D^{M}(\rho\|\sigma)=\sup_{\omega>0}h_{\rho,\sigma}(\omega),\label{eq:measured-rel-ent-var-form}
\end{equation}
where
\begin{equation}
h_{\rho,\sigma}(\omega)\coloneqq\Tr\!\left[\left(\ln\omega\right)\rho\right]+1-\Tr\!\left[\omega\sigma\right].\label{eq:measured-rel-ent-obj-func-def}
\end{equation}
The optimal $\omega$ in \eqref{eq:measured-rel-ent-var-form} leads
to an optimal measurement in \eqref{eq:meas-rel-ent-def}: indeed,
it is known that the optimal measurement in \eqref{eq:meas-rel-ent-def}
is given by a projective measurement in the eigenbasis of the optimal
$\omega$ \cite[Lemma~1 and Theorem~2]{Berta2017}.

We are interested in determining the optimal value of the objective
function $h_{\rho,\sigma}(\omega)$ by means of a gradient ascent
approach. To delineate such an optimization algorithm for this purpose
and as indicated above, we should determine the matrix gradient and
Hessian superoperator of the function $h_{\rho,\sigma}(\omega)$,
as well as its smoothness and concavity properties. After doing so,
we can then plug directly into Algorithm~\ref{alg:NAPGD}.

\subsection{Matrix gradient for measured relative entropy}

Some of the properties listed in this section were recently presented
in~\cite{Sreekumar2025}; however, we list them here for completeness.
\begin{lem}
\label{lem:matrix-gradient-meas-rel-ent}The matrix gradient $\frac{\partial}{\partial\omega}h_{\rho,\sigma}(\omega)$
of $h_{\rho,\sigma}(\omega)$ in \eqref{eq:measured-rel-ent-obj-func-def}
is given by
\begin{align}
\frac{\partial}{\partial\omega}h_{\rho,\sigma}(\omega) & =\int_{0}^{\infty}ds\ \left(\omega+sI\right)^{-1}\rho\left(\omega+sI\right)^{-1}-\sigma\label{eq:matrix-gradient-rel-ent-1}\\
 & =\sum_{\ell,m}f_{\ln x}^{\left[1\right]}(\lambda_{\ell},\lambda_{m})\Pi_{\ell}\rho\Pi_{m}-\sigma,\label{eq:matrix-gradient-rel-ent-2}
\end{align}
where $\omega=\sum_{k}\lambda_{k}\Pi_{k}$ is the spectral decomposition
of $\omega$ and $f_{\ln x}^{\left[1\right]}(x,y)$ is the first divided
difference of the logarithm function, defined for $x,y>0$ as
\begin{equation}
f_{\ln x}^{\left[1\right]}(x,y)\coloneqq\begin{cases}
\frac{1}{x} & :x=y\\
\frac{\ln x-\ln y}{x-y} & :x\neq y
\end{cases}.
\end{equation}
\end{lem}

\begin{proof}
Eqs.~\eqref{eq:matrix-gradient-rel-ent-1}--\eqref{eq:matrix-gradient-rel-ent-2}
follow from a direct application of \eqref{eq:matrix-deriv-n-1-2},
\eqref{eq:matrix-deriv-func-f}, and \eqref{eq:deriv-log}.
\end{proof}
Thus, the first-order optimality condition for the optimal solution
of \eqref{eq:measured-rel-ent-var-form} is that $\omega>0$ should
satisfy
\begin{equation}
0=\frac{\partial}{\partial\omega}h_{\rho,\sigma}(\omega)\quad\implies\quad\int_{0}^{\infty}ds\ \left(\omega+sI\right)^{-1}\rho\left(\omega+sI\right)^{-1}=\sigma,
\end{equation}
or equivalently,
\begin{equation}
\sum_{\ell,m}f_{\ln x}^{\left[1\right]}(\lambda_{\ell},\lambda_{m})\Pi_{\ell}\rho\Pi_{m}=\sigma.\label{eq:first-order-optimal-condition}
\end{equation}
This allows us to determine conditions on the eigenvalues of $\omega$
that are relevant for performing a gradient ascent search for the
optimal $\omega$. Indeed, Lemma~\ref{lem:omega-bounds} below limits
the search space needed in a gradient ascent search algorithm that
searches for the optimal $\omega$, and the bounds are expressed in
terms of quantities related to the max-relative entropy~\cite{Datta2009}.
\begin{lem}
\label{lem:omega-bounds}The optimal $\omega>0$ for $h_{\rho,\sigma}(\omega)$
in \eqref{eq:measured-rel-ent-obj-func-def} satisfies the following:
\begin{equation}
\left\Vert \rho^{-\frac{1}{2}}\sigma\rho^{-\frac{1}{2}}\right\Vert ^{-1}I\leq\omega\leq\left\Vert \sigma^{-\frac{1}{2}}\rho\sigma^{-\frac{1}{2}}\right\Vert I.\label{eq:optimal-omega-meas-rel-ent}
\end{equation}
\end{lem}

\begin{proof}
We apply the map $X\mapsto\Tr[\Pi_{k}X]$ to both sides of \eqref{eq:first-order-optimal-condition}
to find that
\begin{align}
\Tr[\Pi_{k}\sigma] & =\Tr\!\left[\Pi_{k}\left(\sum_{\ell,m}f_{\ln x}^{\left[1\right]}(\lambda_{\ell},\lambda_{m})\Pi_{\ell}\rho\Pi_{m}\right)\right]\\
 & =\sum_{\ell,m}f_{\ln x}^{\left[1\right]}(\lambda_{\ell},\lambda_{m})\Tr\!\left[\Pi_{k}\Pi_{\ell}\rho\Pi_{m}\right]\\
 & =\frac{1}{\lambda_{k}}\Tr\!\left[\Pi_{k}\rho\right],
\end{align}
which implies that every eigenvalue of $\omega$ satisfies
\begin{equation}
\lambda_{k}=\frac{\Tr\!\left[\Pi_{k}\rho\right]}{\Tr\!\left[\Pi_{k}\sigma\right]}.
\end{equation}
Thus, every eigenvalue of $\omega$ satisfies
\begin{equation}
\min_{\Lambda:0\leq\Lambda\leq I}\frac{\Tr\!\left[\Lambda\rho\right]}{\Tr\!\left[\Lambda\sigma\right]}\leq\lambda_{k}=\frac{\Tr\!\left[\Pi_{k}\rho\right]}{\Tr\!\left[\Pi_{k}\sigma\right]}\leq\max_{\Lambda:0\leq\Lambda\leq I}\frac{\Tr\!\left[\Lambda\rho\right]}{\Tr\!\left[\Lambda\sigma\right]}.\label{eq:eigenvalue-bnds-meas-rel-ent}
\end{equation}
Considering that \cite[Lemma~A.4]{Mosonyi2015}
\begin{equation}
\max_{\Lambda:0\leq\Lambda\leq I}\frac{\Tr\!\left[\Lambda\rho\right]}{\Tr\!\left[\Lambda\sigma\right]}=\left\Vert \sigma^{-\frac{1}{2}}\rho\sigma^{-\frac{1}{2}}\right\Vert ,\label{eq:dual-max-rel-entropy}
\end{equation}
and
\begin{align}
\min_{\Lambda:0\leq\Lambda\leq I}\frac{\Tr\!\left[\Lambda\rho\right]}{\Tr\!\left[\Lambda\sigma\right]} & =\min_{\Lambda:0\leq\Lambda\leq I}\left(\frac{\Tr\!\left[\Lambda\sigma\right]}{\Tr\!\left[\Lambda\rho\right]}\right)^{-1}\\
 & =\left(\max_{\Lambda:0\leq\Lambda\leq I}\frac{\Tr\!\left[\Lambda\sigma\right]}{\Tr\!\left[\Lambda\rho\right]}\right)^{-1}\\
 & =\left\Vert \rho^{-\frac{1}{2}}\sigma\rho^{-\frac{1}{2}}\right\Vert ^{-1},\label{eq:dual-flip-max-rel-entropy}
\end{align}
we thus conclude that
\begin{equation}
\left\Vert \rho^{-\frac{1}{2}}\sigma\rho^{-\frac{1}{2}}\right\Vert ^{-1}I\leq\omega\leq\left\Vert \sigma^{-\frac{1}{2}}\rho\sigma^{-\frac{1}{2}}\right\Vert I,
\end{equation}
completing the proof.
\end{proof}
\begin{rem}
We can alternatively express the bounds in Lemma~\ref{lem:omega-bounds}
as follows:
\begin{equation}
\sup_{\mu\geq0}\left\{ \mu:\mu\rho\leq\sigma\right\} I\leq\omega\leq\inf_{\lambda\geq0}\left\{ \lambda:\rho\leq\lambda\sigma\right\} I,
\end{equation}
given that $\left\Vert \sigma^{-\frac{1}{2}}\rho\sigma^{-\frac{1}{2}}\right\Vert =\inf_{\lambda\geq0}\left\{ \lambda:\rho\leq\lambda\sigma\right\} $
and $\left\Vert \rho^{-\frac{1}{2}}\sigma\rho^{-\frac{1}{2}}\right\Vert ^{-1}=\sup_{\mu\geq0}\left\{ \mu:\mu\sigma\leq\rho\right\} $.
\end{rem}

\subsection{Hessian superoperator for measured relative entropy}

To help determine second-order properties of the function $h_{\rho,\sigma}(\omega)$,
such as smoothness and strong concavity, we compute the Hessian superoperator
of $h_{\rho,\sigma}(\omega)$.
\begin{lem}
\label{lem:Hessian-superoperator-measured-rel-ent}The Hessian superoperator
of $h_{\rho,\sigma}(\omega)$ in \eqref{eq:measured-rel-ent-obj-func-def}
is the following self-adjoint, Hermiticity-preserving superoperator:
\begin{equation}
\Phi_{\omega}(X)\coloneqq-\int_{0}^{\infty}ds\ \left(\omega+sI\right)^{-1}\left[X\left(\omega+sI\right)^{-1}\rho+\rho\left(\omega+sI\right)^{-1}X\right]\left(\omega+sI\right)^{-1}.\label{eq:hessian-superoperator}
\end{equation}
\end{lem}

\begin{proof}
Consider that the Hessian superoperator can be computed from
\begin{align}
\frac{\partial}{\partial\omega_{i,j}}\frac{\partial}{\partial\omega}h_{\rho,\sigma}(\omega) & =\frac{\partial}{\partial\omega_{i,j}}\left(\int_{0}^{\infty}ds\ \left(\omega+sI\right)^{-1}\rho\left(\omega+sI\right)^{-1}-\sigma\right)\\
 & =\int_{0}^{\infty}ds\ \frac{\partial}{\partial\omega_{i,j}}\left[\left(\omega+sI\right)^{-1}\rho\left(\omega+sI\right)^{-1}\right],\label{eq:hessian-meas-rel-ent-step-1}
\end{align}
where we applied Lemma~\ref{lem:matrix-gradient-meas-rel-ent} for
the first equality. Now consider that
\begin{multline}
\frac{\partial}{\partial\omega_{i,j}}\left[\left(\omega+sI\right)^{-1}\rho\left(\omega+sI\right)^{-1}\right]=\\
\left[\frac{\partial}{\partial\omega_{i,j}}\left(\omega+sI\right)^{-1}\right]\rho\left(\omega+sI\right)^{-1}+\left(\omega+sI\right)^{-1}\rho\left[\frac{\partial}{\partial\omega_{i,j}}\left(\omega+sI\right)^{-1}\right]\label{eq:hessian-calc-1}
\end{multline}
Recalling that
\begin{equation}
\frac{\partial}{\partial x}A(x)^{-1}=-A(x)^{-1}\left[\frac{\partial}{\partial x}A(x)\right]A(x)^{-1},
\end{equation}
because
\begin{equation}
0=\frac{\partial}{\partial x}I=\frac{\partial}{\partial x}\left(A(x)^{-1}A(x)\right)=\left(\frac{\partial}{\partial x}A(x)^{-1}\right)A(x)+A(x)^{-1}\left(\frac{\partial}{\partial x}A(x)\right),
\end{equation}
it follows that
\begin{align}
\frac{\partial}{\partial\omega_{i,j}}\left(\omega+sI\right)^{-1} & =-\left(\omega+sI\right)^{-1}\left[\frac{\partial}{\partial\omega_{i,j}}\left(\omega+sI\right)\right]\left(\omega+sI\right)^{-1}\\
 & =-\left(\omega+sI\right)^{-1}\left[\frac{\partial}{\partial\omega_{i,j}}\omega\right]\left(\omega+sI\right)^{-1}\\
 & =-\left(\omega+sI\right)^{-1}|i\rangle\!\langle j|\left(\omega+sI\right)^{-1},
\end{align}
where we used the fact that $\frac{\partial}{\partial\omega_{i,j}}\omega=|i\rangle\!\langle j|$
(see \eqref{eq:basic-matrix-deriv-1}--\eqref{eq:basic-matrix-deriv-last}).
Thus,
\begin{align}
 & \left[\frac{\partial}{\partial\omega_{i,j}}\left(\omega+sI\right)^{-1}\right]\rho\left(\omega+sI\right)^{-1}+\left(\omega+sI\right)^{-1}\rho\left[\frac{\partial}{\partial\omega_{i,j}}\left(\omega+sI\right)^{-1}\right]\nonumber \\
 & =-\left(\omega+sI\right)^{-1}|i\rangle\!\langle j|\left(\omega+sI\right)^{-1}\rho\left(\omega+sI\right)^{-1}\nonumber \\
 & \qquad-\left(\omega+sI\right)^{-1}\rho\left(\omega+sI\right)^{-1}|i\rangle\!\langle j|\left(\omega+sI\right)^{-1}\\
 & =-\left(\omega+sI\right)^{-1}\left[|i\rangle\!\langle j|\left(\omega+sI\right)^{-1}\rho+\rho\left(\omega+sI\right)^{-1}|i\rangle\!\langle j|\right]\left(\omega+sI\right)^{-1},\label{eq:hessian-calc-final}
\end{align}
implying that the Hessian superoperator is given by \eqref{eq:hessian-superoperator},
after combining with \eqref{eq:hessian-meas-rel-ent-step-1}.

The superoperator in \eqref{eq:hessian-superoperator} is Hermiticity
preserving by inspection.

Defining $\Phi_{s,\rho,\omega}(X)$ as
\begin{equation}
\Phi_{s,\rho,\omega}(X)\coloneqq\left(\omega+sI\right)^{-1}\left[X\left(\omega+sI\right)^{-1}\rho+\rho\left(\omega+sI\right)^{-1}X\right]\left(\omega+sI\right)^{-1},\label{eq:Phi-s-rho-superop}
\end{equation}
and noting that
\begin{equation}
\Phi_{\omega}(X)=-\int_{0}^{\infty}ds\ \Phi_{s,\rho,\omega}(X),
\end{equation}
consider that
\begin{align}
\left\langle Y,\Phi_{s,\rho,\omega}(X)\right\rangle  & =\left\langle Y,\left(\omega+sI\right)^{-1}\left[X\left(\omega+sI\right)^{-1}\rho+\rho\left(\omega+sI\right)^{-1}X\right]\left(\omega+sI\right)^{-1}\right\rangle \label{eq:self-adjoint-hessian-superop-1}\\
 & =\left\langle Y,\left(\omega+sI\right)^{-1}X\left(\omega+sI\right)^{-1}\rho\left(\omega+sI\right)^{-1}\right\rangle \nonumber \\
 & \qquad+\left\langle Y,\left(\omega+sI\right)^{-1}\rho\left(\omega+sI\right)^{-1}X\left(\omega+sI\right)^{-1}\right\rangle \\
 & =\Tr\!\left[Y^{\dag}\left(\omega+sI\right)^{-1}X\left(\omega+sI\right)^{-1}\rho\left(\omega+sI\right)^{-1}\right]\nonumber \\
 & \qquad+\Tr\!\left[Y^{\dag}\left(\omega+sI\right)^{-1}\rho\left(\omega+sI\right)^{-1}X\left(\omega+sI\right)^{-1}\right]\\
 & =\Tr\!\left[\left(\omega+sI\right)^{-1}\rho\left(\omega+sI\right)^{-1}Y^{\dag}\left(\omega+sI\right)^{-1}X\right]\nonumber \\
 & \qquad+\Tr\!\left[\left(\omega+sI\right)^{-1}Y^{\dag}\left(\omega+sI\right)^{-1}\rho\left(\omega+sI\right)^{-1}X\right]\\
 & =\Tr\!\left[\left[\left(\omega+sI\right)^{-1}Y\left(\omega+sI\right)^{-1}\rho\left(\omega+sI\right)^{-1}\right]^{\dag}X\right]\nonumber \\
 & \qquad+\Tr\!\left[\left[\left(\omega+sI\right)^{-1}\rho\left(\omega+sI\right)^{-1}Y\left(\omega+sI\right)^{-1}\right]^{\dag}X\right]\\
 & =\left\langle \Phi_{s,\rho,\omega}(Y),X\right\rangle .\label{eq:self-adjoint-hessian-superop-last}
\end{align}
It follows that $\Phi_{\omega}(X)$ is self-adjoint after observing
that
\begin{align}
\left\langle Y,\Phi_{\omega}(X)\right\rangle  & =-\int_{0}^{\infty}ds\ \left\langle Y,\Phi_{s,\rho,\omega}(X)\right\rangle \\
 & =-\int_{0}^{\infty}ds\ \left\langle \Phi_{s,\rho,\omega}(Y),X\right\rangle \\
 & =\left\langle \Phi_{\omega}(Y),X\right\rangle ,
\end{align}
thus concluding the proof.
\end{proof}

\subsection{Smoothness and strong concavity for measured relative entropy}
\begin{lem}
For $\omega>0$ satisfying \eqref{eq:optimal-omega-meas-rel-ent},
the Hessian superoperator $\Phi_{\omega}$ in \eqref{eq:hessian-superoperator}
has its minimum and maximum eigenvalues bounded as follows:
\begin{equation}
-\lambda_{\max}(\rho)\left\Vert \rho^{-\frac{1}{2}}\sigma\rho^{-\frac{1}{2}}\right\Vert ^{2}\leq\min_{Y:\left\Vert Y\right\Vert _{2}=1}\left\langle Y,\Phi_{\omega}(Y)\right\rangle \leq\max_{Y:\left\Vert Y\right\Vert _{2}=1}\left\langle Y,\Phi_{\omega}(Y)\right\rangle \leq-\frac{\lambda_{\min}(\rho)}{\left\Vert \sigma^{-\frac{1}{2}}\rho\sigma^{-\frac{1}{2}}\right\Vert ^{2}}.
\end{equation}
\end{lem}

\begin{proof}
We begin by bounding the maximum and minimum eigenvalues of the superoperator~$\Phi_{s,\rho,\omega}$
defined in \eqref{eq:Phi-s-rho-superop}. Defining
\begin{equation}
u\equiv\left\Vert \sigma^{-\frac{1}{2}}\rho\sigma^{-\frac{1}{2}}\right\Vert ,
\end{equation}
and letting $Y$ satisfy $\left\Vert Y\right\Vert _{2}=1$, the minimum
eigenvalue of $\Phi_{s,\rho,\omega}$ can be bounded from below as
follows:
\begin{align}
\left\langle Y,\Phi_{s,\rho,\omega}(Y)\right\rangle  & =\Tr\!\left[Y^{\dag}\left(\omega+sI\right)^{-1}Y\left(\omega+sI\right)^{-1}\rho\left(\omega+sI\right)^{-1}\right]\nonumber \\
 & \qquad+\Tr\!\left[Y^{\dag}\left(\omega+sI\right)^{-1}\rho\left(\omega+sI\right)^{-1}Y\left(\omega+sI\right)^{-1}\right]\label{eq:Phi-s-rho-min-eig-1}\\
 & =\Tr\!\left[Y^{\dag}\left(\omega+sI\right)^{-1}Y\left(\omega+sI\right)^{-1}\rho\left(\omega+sI\right)^{-1}\right]\nonumber \\
 & \qquad+\Tr\!\left[Y\left(\omega+sI\right)^{-1}Y^{\dag}\left(\omega+sI\right)^{-1}\rho\left(\omega+sI\right)^{-1}\right]\\
 & \geq\Tr\!\left[Y^{\dag}\left(\left(u+s\right)^{-1}I\right)Y\left(\omega+sI\right)^{-1}\rho\left(\omega+sI\right)^{-1}\right]\nonumber \\
 & \qquad+\Tr\!\left[Y\left(\left(u+s\right)^{-1}I\right)Y^{\dag}\left(\omega+sI\right)^{-1}\rho\left(\omega+sI\right)^{-1}\right]\\
 & =\left(u+s\right)^{-1}\Tr\!\left[\left(Y^{\dag}Y+YY^{\dag}\right)\left(\omega+sI\right)^{-1}\rho\left(\omega+sI\right)^{-1}\right]\\
 & \geq\left(u+s\right)^{-1}\Tr\!\left[\left(Y^{\dag}Y+YY^{\dag}\right)\left(\omega+sI\right)^{-1}\left(\lambda_{\min}(\rho)I\right)\left(\omega+sI\right)^{-1}\right]\\
 & =\lambda_{\min}(\rho)\left(u+s\right)^{-1}\Tr\!\left[\left(Y^{\dag}Y+YY^{\dag}\right)\left(\omega+sI\right)^{-2}\right]\\
 & \geq\lambda_{\min}(\rho)\left(u+s\right)^{-3}\Tr\!\left[Y^{\dag}Y+YY^{\dag}\right]\\
 & =2\lambda_{\min}(\rho)\left(u+s\right)^{-3}.\label{eq:Phi-s-rho-min-eig-last}
\end{align}
For the first inequality, we used the fact that $\left(\omega+sI\right)^{-1}\geq\left(u+s\right)^{-1}I$,
which follows from the assumption that $\omega$ satisfies \eqref{eq:optimal-omega-meas-rel-ent}.
For the penultimate quality, we used that $\Tr\!\left[Y^{\dag}Y+YY^{\dag}\right]=2$
because $\left\Vert Y\right\Vert _{2}=1$ by assumption.

Defining
\begin{equation}
\ell\equiv\left\Vert \rho^{-\frac{1}{2}}\sigma\rho^{-\frac{1}{2}}\right\Vert ^{-1}
\end{equation}
and letting $Y$ satisfy $\left\Vert Y\right\Vert _{2}=1$, the maximum
eigenvalue of $\Phi_{s,\rho,\omega}$ can be bounded from above as
follows: 
\begin{align}
\left\langle Y,\Phi_{s,\rho,\omega}(Y)\right\rangle  & =\Tr\!\left[Y^{\dag}\left(\omega+sI\right)^{-1}Y\left(\omega+sI\right)^{-1}\rho\left(\omega+sI\right)^{-1}\right]\nonumber \\
 & \qquad+\Tr\!\left[Y\left(\omega+sI\right)^{-1}Y^{\dag}\left(\omega+sI\right)^{-1}\rho\left(\omega+sI\right)^{-1}\right]\label{eq:Phi-s-rho-max-eig-1}\\
 & \leq\Tr\!\left[Y^{\dag}\left(\left(\ell+s\right)I\right)^{-1}Y\left(\omega+sI\right)^{-1}\rho\left(\omega+sI\right)^{-1}\right]\nonumber \\
 & \qquad+\Tr\!\left[Y\left(\left(\ell+s\right)I\right)^{-1}Y^{\dag}\left(\omega+sI\right)^{-1}\rho\left(\omega+sI\right)^{-1}\right]\\
 & =\left(\ell+s\right)^{-1}\Tr\!\left[\left(Y^{\dag}Y+YY^{\dag}\right)\left(\omega+sI\right)^{-1}\rho\left(\omega+sI\right)^{-1}\right]\\
 & \leq\left(\ell+s\right)^{-1}\Tr\!\left[\left(Y^{\dag}Y+YY^{\dag}\right)\left(\omega+sI\right)^{-1}\left(\lambda_{\max}(\rho)I\right)\left(\omega+sI\right)^{-1}\right]\\
 & =\lambda_{\max}(\rho)\left(\ell+s\right)^{-1}\Tr\!\left[\left(Y^{\dag}Y+YY^{\dag}\right)\left(\omega+sI\right)^{-2}\right]\\
 & \leq\lambda_{\max}(\rho)\left(\ell+s\right)^{-3}\Tr\!\left[Y^{\dag}Y+YY^{\dag}\right]\\
 & =2\lambda_{\max}(\rho)\left(\ell+s\right)^{-3},\label{eq:Phi-s-rho-max-eig-last}
\end{align}
where we followed a similar line of reasoning and again invoked the
assumption that $\omega$ satisfies \eqref{eq:optimal-omega-meas-rel-ent}.

We now apply these bounds to $\Phi_{\omega}$, by observing that 
\begin{align}
\left\langle Y,\Phi_{\omega}(Y)\right\rangle  & =-\int_{0}^{\infty}ds\ \left\langle Y,\Phi_{s,\rho,\omega}(Y)\right\rangle \\
 & \leq-2\lambda_{\min}(\rho)\int_{0}^{\infty}ds\ \left(u+s\right)^{-3}\\
 & =-\frac{2\lambda_{\min}(\rho)}{2u^{2}}\\
 & =-\frac{\lambda_{\min}(\rho)}{\left\Vert \sigma^{-\frac{1}{2}}\rho\sigma^{-\frac{1}{2}}\right\Vert ^{2}}.
\end{align}
and
\begin{align}
\left\langle Y,\Phi_{\omega}(Y)\right\rangle  & =-\int_{0}^{\infty}ds\ \left\langle Y,\Phi_{s,\rho,\omega}(Y)\right\rangle \\
 & \geq-2\lambda_{\max}(\rho)\int_{0}^{\infty}ds\ \left(\ell+s\right)^{-3}\\
 & =-\frac{2\lambda_{\max}(\rho)}{2\ell^{2}}\\
 & =-\lambda_{\max}(\rho)\left\Vert \rho^{-\frac{1}{2}}\sigma\rho^{-\frac{1}{2}}\right\Vert ^{2},
\end{align}
thus concluding the proof.
\end{proof}
\begin{cor}
\label{cor:smoothness-strong-concavity-meas-rel-ent}On the operator
interval
\begin{equation}
\omega\in\left[\left\Vert \rho^{-\frac{1}{2}}\sigma\rho^{-\frac{1}{2}}\right\Vert ^{-1}I,\left\Vert \sigma^{-\frac{1}{2}}\rho\sigma^{-\frac{1}{2}}\right\Vert I\right],
\end{equation}
the function $h_{\rho,\sigma}(\omega)$ is $\beta$-smooth with parameter
\begin{equation}
\beta\coloneqq\lambda_{\max}(\rho)\left\Vert \rho^{-\frac{1}{2}}\sigma\rho^{-\frac{1}{2}}\right\Vert ^{2},
\end{equation}
and it is $\gamma$-strongly concave with parameter
\begin{equation}
\gamma\coloneqq\frac{\lambda_{\min}(\rho)}{\left\Vert \sigma^{-\frac{1}{2}}\rho\sigma^{-\frac{1}{2}}\right\Vert ^{2}}.
\end{equation}
Thus, on this operator interval, the condition number $\kappa$ of
the function $h_{\rho,\sigma}(\omega)$ is given by
\begin{equation}
\kappa\coloneqq\frac{\beta}{\gamma}=\frac{\lambda_{\max}(\rho)}{\lambda_{\min}(\rho)}\left(\left\Vert \rho^{-\frac{1}{2}}\sigma\rho^{-\frac{1}{2}}\right\Vert \left\Vert \sigma^{-\frac{1}{2}}\rho\sigma^{-\frac{1}{2}}\right\Vert \right)^{2}.\label{eq:kappa-meas-rel-ent}
\end{equation}
\end{cor}

\begin{proof}
This is an immediate consequence of the eigenvalue bounds from Lemma
\ref{lem:Hessian-superoperator-measured-rel-ent}.
\end{proof}

\subsection{Nesterov accelerated optimization of measured relative entropy}

We employ Nesterov accelerated projected gradient ascent in order
to perform the optimization in \eqref{eq:measured-rel-ent-var-form}.
\begin{lyxalgorithm}
\label{alg:measured-rel-ent}Proceed according to the following steps:
\begin{enumerate}
\item Set $m=0$, and initialize
\begin{equation}
\omega_{0}\leftarrow\xi_{0}\leftarrow\frac{1}{2}\left(\left\Vert \rho^{-\frac{1}{2}}\sigma\rho^{-\frac{1}{2}}\right\Vert ^{-1}+\left\Vert \sigma^{-\frac{1}{2}}\rho\sigma^{-\frac{1}{2}}\right\Vert \right)I.
\end{equation}
\item While $\left\Vert \left.\frac{\partial}{\partial\omega}h_{\rho,\sigma}(\omega)\right|_{\omega=\omega_{m}}\right\Vert _{2}>\sqrt{2\gamma\varepsilon}$,
\begin{align}
\xi_{m+1} & \leftarrow\omega_{m}+\frac{1}{\beta}\left.\frac{\partial}{\partial\omega}h_{\rho,\sigma}(\omega)\right|_{\omega=\omega_{m}},\\
\chi_{m+1} & \leftarrow\left[\xi_{m+1}\right]_{\mathcal{I}},\\
\omega_{m+1} & \leftarrow\chi_{m+1}+\left(\frac{\sqrt{\kappa}-1}{\sqrt{\kappa}+1}\right)\left(\chi_{m+1}-\chi_{m}\right),
\end{align}
where
\begin{equation}
\mathcal{I}\coloneqq\left[\left\Vert \rho^{-\frac{1}{2}}\sigma\rho^{-\frac{1}{2}}\right\Vert ^{-1},\left\Vert \sigma^{-\frac{1}{2}}\rho\sigma^{-\frac{1}{2}}\right\Vert \right],
\end{equation}
and $\gamma$, $\beta$, and $\kappa$ are defined in Corollary~\ref{cor:smoothness-strong-concavity-meas-rel-ent}.
Set $m=m+1$.
\item Output $h_{\rho,\sigma}(\omega_{m})$ as an estimate of $D^{M}(\rho\|\sigma)$
satisfying \eqref{eq:meas-rel-ent-eps-error}.
\end{enumerate}
\end{lyxalgorithm}

The stopping condition $\left\Vert \left.\frac{\partial}{\partial\omega}h_{\rho,\sigma}(\omega)\right|_{\omega=\omega_{m}}\right\Vert _{2}\leq\sqrt{2\gamma\varepsilon}$
implies that 
\begin{equation}
\left|D^{M}(\rho\|\sigma)-h_{\rho,\sigma}(\omega_{m})\right|\leq\varepsilon,\label{eq:meas-rel-ent-eps-error}
\end{equation}
as a consequence of the same reasoning that led to \eqref{eq:error-threshold-desired}.
The number of iterations required by the above algorithm to guarantee
\eqref{eq:meas-rel-ent-eps-error} is equal to \cite[Theorem~3.18]{Bubeck2015}
\begin{equation}
O\!\left(\sqrt{\kappa}\ln\!\left(\frac{1}{\varepsilon}\right)\right).
\end{equation}

\section{Measured R\'enyi relative entropy}

\label{sec:Measured-Renyi-relative}For a positive definite state
$\rho$, a positive definite operator $\sigma$, and $\alpha\in\left(0,1\right)\cup\left(1,\infty\right)$,
the measured R\'enyi relative entropy is defined as \cite[Eqs.~(3.116)--(3.117)]{Fuchs1996}
\begin{equation}
D_{\alpha}^{M}(\rho\|\sigma)\coloneqq\frac{1}{\alpha-1}\ln Q_{\alpha}^{M}(\rho\|\sigma),\label{eq:measured-renyi-rel-ent-def}
\end{equation}
where
\begin{equation}
Q_{\alpha}^{M}(\rho\|\sigma)\coloneqq\begin{cases}
\inf_{\left(\Lambda_{x}\right)_{x\in\mathcal{X}}}\sum_{x\in\mathcal{X}}\left(\Tr\!\left[\Lambda_{x}\rho\right]\right)^{\alpha}\left(\Tr\!\left[\Lambda_{x}\sigma\right]\right)^{1-\alpha} & :\alpha\in\left(0,1\right)\\
\sup_{\left(\Lambda_{x}\right)_{x\in\mathcal{X}}}\sum_{x\in\mathcal{X}}\left(\Tr\!\left[\Lambda_{x}\rho\right]\right)^{\alpha}\left(\Tr\!\left[\Lambda_{x}\sigma\right]\right)^{1-\alpha} & :\alpha\in\left(1,\infty\right)
\end{cases}.
\end{equation}
As for measured relative entropy, $\left(\Lambda_{x}\right)_{x\in\mathcal{X}}$
is a positive operator-valued measure (i.e., satisfying $\Lambda_{x}\geq0$
for all $x\in\mathcal{X}$ and $\sum_{x\in\mathcal{X}}\Lambda_{x}=I$).
The latter quantity $Q_{\alpha}^{M}(\rho\|\sigma)$ is called measured
R\'enyi relative quasi-entropy, and it can be expressed in the following
variational form \cite[Lemma~3]{Berta2017}:
\begin{equation}
Q_{\alpha}^{M}(\rho\|\sigma)=\begin{cases}
\inf_{\omega>0}g_{\rho,\sigma}^{\alpha}(\omega) & :\alpha\in\left(0,\frac{1}{2}\right)\\
\inf_{\omega>0}h_{\rho,\sigma}^{\alpha}(\omega) & :\alpha\in\left[\frac{1}{2},1\right)\\
\sup_{\omega>0}h_{\rho,\sigma}^{\alpha}(\omega) & :\alpha\in\left(1,\infty\right)
\end{cases},\label{eq:meas-renyi-var-form}
\end{equation}
where
\begin{align}
g_{\rho,\sigma}^{\alpha}(\omega) & \coloneqq\alpha\Tr[\omega\rho]+\left(1-\alpha\right)\Tr\!\left[\omega^{\frac{\alpha}{\alpha-1}}\sigma\right],\label{eq:g-alpha-meas-renyi}\\
h_{\rho,\sigma}^{\alpha}(\omega) & \coloneqq\alpha\Tr\!\left[\omega^{\frac{\alpha-1}{\alpha}}\rho\right]+\left(1-\alpha\right)\Tr[\omega\sigma].\label{eq:h-alpha-meas-renyi}
\end{align}
Observe that $g_{\rho,\sigma}^{\alpha}(\omega)=h_{\rho,\sigma}^{\alpha}(\omega)$
for $\alpha=\frac{1}{2}$.

\subsection{Matrix gradient of measured R\'enyi relative entropy}

Some of the properties listed in this section were recently presented
in~\cite{Sreekumar2025}; however, we list them here for completeness.
\begin{lem}
\label{lem:matrix-gradient-meas-renyi-rel-ent}The matrix gradient
$\frac{\partial}{\partial\omega}g_{\rho,\sigma}^{\alpha}(\omega)$
for $\alpha\in\left(0,\frac{1}{2}\right)$ is given by
\begin{align}
\frac{\partial}{\partial\omega}g_{\rho,\sigma}^{\alpha}(\omega) & =\alpha\rho+\left(1-\alpha\right)\sum_{\ell,m}f_{x^{\frac{\alpha}{\alpha-1}}}^{[1]}(\lambda_{\ell},\lambda_{m})\Pi_{\ell}\sigma\Pi_{m}\label{eq:matrix-gradient-g-divided-diff}\\
 & =\alpha\rho+c_{1}(\alpha)\int_{0}^{\infty}dt\ t^{\frac{\alpha}{\alpha-1}}\left(\omega+tI\right)^{-1}\sigma\left(\omega+tI\right)^{-1},\label{eq:matrix-gradient-g}\\
c_{1}(\alpha) & \equiv\left(1-\alpha\right)\frac{\sin\!\left(\left(\frac{\alpha}{\alpha-1}\right)\pi\right)}{\pi},\label{eq:c1-alpha-constant}
\end{align}
and the matrix gradient $\frac{\partial}{\partial\omega}h_{\rho,\sigma}^{\alpha}(\omega)$
for $\alpha\in\left[\frac{1}{2},1\right)\cup\left(1,\infty\right)$
is given by
\begin{align}
\frac{\partial}{\partial\omega}h_{\rho,\sigma}^{\alpha}(\omega) & =\left(1-\alpha\right)\sigma+\alpha\sum_{\ell,m}f_{x^{\frac{\alpha-1}{\alpha}}}^{[1]}(\lambda_{\ell},\lambda_{m})\Pi_{\ell}\rho\Pi_{m}\label{eq:matrix-gradient-h-divided-diff}\\
 & =\left(1-\alpha\right)\sigma+c_{2}(\alpha)\int_{0}^{\infty}dt\ t^{\frac{\alpha-1}{\alpha}}\left(\omega+tI\right)^{-1}\rho\left(\omega+tI\right)^{-1},\label{eq:matrix-gradient-h}\\
c_{2}(\alpha) & \equiv\alpha\frac{\sin\!\left(\left(\frac{\alpha-1}{\alpha}\right)\pi\right)}{\pi},\label{eq:c2-alpha-constant}
\end{align}
where $\omega=\sum_{k}\lambda_{k}\Pi_{k}$ is the spectral decomposition
of $\omega$. Furthermore, $f_{x^{\frac{\alpha}{\alpha-1}}}^{[1]}(x,y)$
and $f_{x^{\frac{\alpha-1}{\alpha}}}^{[1]}(x,y)$ in \eqref{eq:matrix-gradient-g-divided-diff}
and \eqref{eq:matrix-gradient-h-divided-diff} are first divided difference
functions of $x\mapsto x^{\frac{\alpha}{\alpha-1}}$ and $x\mapsto x^{\frac{\alpha-1}{\alpha}}$,
respectively, defined for $x,y>0$ as
\begin{align}
f_{x^{\frac{\alpha}{\alpha-1}}}^{[1]}(x,y) & \coloneqq\begin{cases}
\left(\frac{\alpha}{\alpha-1}\right)x^{\frac{1}{\alpha-1}} & :x=y\\
\frac{x^{\frac{\alpha}{\alpha-1}}-y^{\frac{\alpha}{\alpha-1}}}{x-y} & :x\neq y\ ,
\end{cases}\\
f_{x^{\frac{\alpha-1}{\alpha}}}^{[1]}(x,y) & \coloneqq\begin{cases}
\left(\frac{\alpha-1}{\alpha}\right)x^{-\frac{1}{\alpha}} & :x=y\\
\frac{x^{\frac{\alpha-1}{\alpha}}-y^{\frac{\alpha-1}{\alpha}}}{x-y} & :x\neq y\ .
\end{cases}
\end{align}
\end{lem}

\begin{proof}
Eqs.~\eqref{eq:matrix-gradient-g-divided-diff}--\eqref{eq:c2-alpha-constant}
follow from a direct application of \eqref{eq:matrix-deriv-n-1-2},
\eqref{eq:matrix-deriv-func-f}, and \eqref{eq:power-func-deriv-minus-1-to-plus-1}.
In more detail, for $\alpha\in\left(0,\frac{1}{2}\right)$, consider
that
\begin{align}
\frac{\partial}{\partial\omega}g_{\rho,\sigma}^{\alpha}(\omega) & =\frac{\partial}{\partial\omega}\left(\alpha\Tr[\rho\omega]+\left(1-\alpha\right)\Tr[\sigma\omega^{\frac{\alpha}{\alpha-1}}]\right)\\
 & =\alpha\frac{\partial}{\partial\omega}\Tr[\rho\omega]+\left(1-\alpha\right)\frac{\partial}{\partial\omega}\Tr[\sigma\omega^{\frac{\alpha}{\alpha-1}}]\\
 & =\alpha\rho+\left(1-\alpha\right)\sum_{\ell,m}f_{x^{\frac{\alpha}{\alpha-1}}}^{[1]}(\lambda_{\ell},\lambda_{m})\Pi_{\ell}\sigma\Pi_{m}\\
 & =\alpha\rho+c_{1}(\alpha)\int_{0}^{\infty}dt\ t^{\frac{\alpha}{\alpha-1}}\left(\omega+tI\right)^{-1}\sigma\left(\omega+tI\right)^{-1}.
\end{align}
Additionally, for $\alpha\in\left[\frac{1}{2},1\right)\cup\left(1,\infty\right)$,
consider that
\begin{align}
\frac{\partial}{\partial\omega}h_{\rho,\sigma}^{\alpha}(\omega) & =\frac{\partial}{\partial\omega}\left(\alpha\Tr[\rho\omega^{\frac{\alpha-1}{\alpha}}]+\left(1-\alpha\right)\Tr[\sigma\omega]\right)\\
 & =\alpha\frac{\partial}{\partial\omega}\Tr[\rho\omega^{\frac{\alpha-1}{\alpha}}]+\left(1-\alpha\right)\frac{\partial}{\partial\omega}\Tr[\sigma\omega]\\
 & =\left(1-\alpha\right)\sigma+\alpha\sum_{\ell,m}f_{x^{\frac{\alpha-1}{\alpha}}}^{[1]}(\lambda_{\ell},\lambda_{m})\Pi_{\ell}\rho\Pi_{m}\\
 & =\left(1-\alpha\right)\sigma+c_{2}(\alpha)\int_{0}^{\infty}dt\ t^{\frac{\alpha-1}{\alpha}}\left(\omega+tI\right)^{-1}\rho\left(\omega+tI\right)^{-1},
\end{align}
thus concluding the proof.
\end{proof}
Thus, the first-order optimality conditions for the optimal solutions
of \eqref{eq:meas-renyi-var-form} are that $\omega>0$ should satisfy,
for $\alpha\in\left(0,\frac{1}{2}\right)$,
\begin{equation}
0=\frac{\partial}{\partial\omega}g_{\rho,\sigma}^{\alpha}(\omega)\quad\implies\quad\alpha\rho+c_{1}(\alpha)\int_{0}^{\infty}dt\ t^{\frac{\alpha}{\alpha-1}}\left(\omega+tI\right)^{-1}\sigma\left(\omega+tI\right)^{-1}=0,
\end{equation}
or equivalently,
\begin{equation}
\alpha\rho+\left(1-\alpha\right)\sum_{\ell,m}f_{x^{\frac{\alpha}{\alpha-1}}}^{[1]}(\lambda_{\ell},\lambda_{m})\Pi_{\ell}\sigma\Pi_{m}=0,\label{eq:first-order-optimal-condition-renyi-0-half}
\end{equation}
and for $\alpha\in\left[\frac{1}{2},1\right)\cup\left(1,\infty\right)$,
\begin{equation}
0=\frac{\partial}{\partial\omega}h_{\rho,\sigma}^{\alpha}(\omega)\quad\implies\quad\left(1-\alpha\right)\sigma+c_{2}(\alpha)\int_{0}^{\infty}dt\ t^{\frac{\alpha-1}{\alpha}}\left(\omega+tI\right)^{-1}\rho\left(\omega+tI\right)^{-1}=0,
\end{equation}
or equivalently,
\begin{equation}
\left(1-\alpha\right)\sigma+\alpha\sum_{\ell,m}f_{x^{\frac{\alpha-1}{\alpha}}}^{[1]}(\lambda_{\ell},\lambda_{m})\Pi_{\ell}\rho\Pi_{m}=0.\label{eq:first-order-optimal-condition-renyi-half-infty}
\end{equation}
This allows us to determine conditions on the eigenvalues of $\omega$,
which are relevant for performing a search for the optimal $\omega$.
\begin{lem}
\label{lem:omega-bounds-1}The optimal $\omega>0$ for $g_{\rho,\sigma}^{\alpha}(\omega)$
and $\alpha\in\left(0,\frac{1}{2}\right)$ satisfies the following:
\begin{equation}
\left\Vert \sigma^{-\frac{1}{2}}\rho\sigma^{-\frac{1}{2}}\right\Vert ^{-\left(1-\alpha\right)}I\leq\omega\leq\left\Vert \rho^{-\frac{1}{2}}\sigma\rho^{-\frac{1}{2}}\right\Vert ^{1-\alpha}I.\label{eq:optimal-omega-alpha-0-half}
\end{equation}
The optimal $\omega>0$ for $h_{\rho,\sigma}^{\alpha}(\omega)$ and
$\alpha\in\left[\frac{1}{2},1\right)\cup\left(1,\infty\right)$ satisfies
the following:
\begin{equation}
\left\Vert \rho^{-\frac{1}{2}}\sigma\rho^{-\frac{1}{2}}\right\Vert ^{-\alpha}I\leq\omega\leq\left\Vert \sigma^{-\frac{1}{2}}\rho\sigma^{-\frac{1}{2}}\right\Vert ^{\alpha}I.\label{eq:optimal-omega-alpha-half-infty}
\end{equation}
\end{lem}

\begin{proof}
We begin with $g_{\rho,\sigma}^{\alpha}(\omega)$ and $\alpha\in\left(0,\frac{1}{2}\right)$
and then apply the map $X\mapsto\Tr[\Pi_{k}X]$ to \eqref{eq:first-order-optimal-condition-renyi-0-half}
to find that
\begin{align}
0 & =\Tr\!\left[\Pi_{k}\left(\alpha\rho+\left(1-\alpha\right)\sum_{\ell,m}f_{x^{\frac{\alpha}{\alpha-1}}}^{[1]}(\lambda_{\ell},\lambda_{m})\Pi_{\ell}\sigma\Pi_{m}\right)\right]\\
 & =\alpha\Tr\!\left[\Pi_{k}\rho\right]+\left(1-\alpha\right)\Tr\!\left[\Pi_{k}\left(\sum_{\ell,m}f_{x^{\frac{\alpha}{\alpha-1}}}^{[1]}(\lambda_{\ell},\lambda_{m})\Pi_{\ell}\sigma\Pi_{m}\right)\right]\\
 & =\alpha\Tr\!\left[\Pi_{k}\rho\right]+\left(1-\alpha\right)\sum_{\ell,m}f_{x^{\frac{\alpha}{\alpha-1}}}^{[1]}(\lambda_{\ell},\lambda_{m})\Tr\!\left[\Pi_{k}\Pi_{\ell}\sigma\Pi_{m}\right]\\
 & =\alpha\Tr\!\left[\Pi_{k}\rho\right]+\left(1-\alpha\right)\left(\frac{\alpha}{\alpha-1}\right)\lambda_{k}^{\frac{1}{\alpha-1}}\Tr\!\left[\Pi_{k}\sigma\right],
\end{align}
which implies that
\begin{equation}
\frac{\Tr\!\left[\Pi_{k}\rho\right]}{\Tr\!\left[\Pi_{k}\sigma\right]}=\lambda_{k}^{\frac{1}{\alpha-1}},
\end{equation}
and is equivalent to
\begin{equation}
\lambda_{k}=\left(\frac{\Tr\!\left[\Pi_{k}\sigma\right]}{\Tr\!\left[\Pi_{k}\rho\right]}\right)^{1-\alpha}.
\end{equation}
Thus, every eigenvalue satisfies
\begin{align}
\left\Vert \sigma^{-\frac{1}{2}}\rho\sigma^{-\frac{1}{2}}\right\Vert ^{-\left(1-\alpha\right)} & =\left(\min_{\Lambda:0\leq\Lambda\leq I}\frac{\Tr\!\left[\Lambda\sigma\right]}{\Tr\!\left[\Lambda\rho\right]}\right)^{1-\alpha}\\
 & \leq\lambda_{k}=\left(\frac{\Tr\!\left[\Pi_{k}\sigma\right]}{\Tr\!\left[\Pi_{k}\rho\right]}\right)^{1-\alpha}\\
 & \leq\left(\max_{\Lambda:0\leq\Lambda\leq I}\frac{\Tr\!\left[\Lambda\sigma\right]}{\Tr\!\left[\Lambda\rho\right]}\right)^{1-\alpha}\\
 & =\left\Vert \rho^{-\frac{1}{2}}\sigma\rho^{-\frac{1}{2}}\right\Vert ^{1-\alpha},
\end{align}
where we applied \eqref{eq:dual-max-rel-entropy} and \eqref{eq:dual-flip-max-rel-entropy},
as well as monotonicity of the function $x\mapsto x^{1-\alpha}$ for
$\alpha\in\left(0,\frac{1}{2}\right)$. We thus conclude \eqref{eq:optimal-omega-alpha-0-half}.

Now let us consider $h_{\rho,\sigma}^{\alpha}(\omega)$ and $\alpha\in\left[\frac{1}{2},1\right)\cup\left(1,\infty\right)$.
Apply the map $X\mapsto\Tr[\Pi_{k}X]$ to \eqref{eq:first-order-optimal-condition-renyi-half-infty}
to find that
\begin{align}
0 & =\Tr\!\left[\Pi_{k}\left(\left(1-\alpha\right)\sigma+\alpha\sum_{\ell,m}f_{x^{\frac{\alpha-1}{\alpha}}}^{[1]}(\lambda_{\ell},\lambda_{m})\Pi_{\ell}\rho\Pi_{m}\right)\right]\\
 & =\left(1-\alpha\right)\Tr\!\left[\Pi_{k}\sigma\right]+\alpha\Tr\!\left[\Pi_{k}\left(\sum_{\ell,m}f_{x^{\frac{\alpha-1}{\alpha}}}^{[1]}(\lambda_{\ell},\lambda_{m})\Pi_{\ell}\rho\Pi_{m}\right)\right]\\
 & =\left(1-\alpha\right)\Tr\!\left[\Pi_{k}\sigma\right]+\alpha\sum_{\ell,m}f_{x^{\frac{\alpha-1}{\alpha}}}^{[1]}(\lambda_{\ell},\lambda_{m})\Tr\!\left[\Pi_{k}\Pi_{\ell}\rho\Pi_{m}\right]\\
 & =\left(1-\alpha\right)\Tr\!\left[\Pi_{k}\sigma\right]+\alpha\left(\frac{\alpha-1}{\alpha}\right)\lambda_{k}^{-\frac{1}{\alpha}}\Tr\!\left[\Pi_{k}\rho\right],
\end{align}
which implies that
\begin{equation}
\lambda_{k}^{-\frac{1}{\alpha}}\Tr\!\left[\Pi_{k}\rho\right]=\Tr\!\left[\Pi_{k}\sigma\right],
\end{equation}
and is equivalent to
\begin{equation}
\lambda_{k}=\left(\frac{\Tr\!\left[\Pi_{k}\rho\right]}{\Tr\!\left[\Pi_{k}\sigma\right]}\right)^{\alpha}.
\end{equation}
Then proceeding as in \eqref{eq:eigenvalue-bnds-meas-rel-ent}--\eqref{eq:dual-flip-max-rel-entropy}
(but with the $\alpha$ power included), we conclude \eqref{eq:optimal-omega-alpha-half-infty}.
\end{proof}

\subsection{Hessian superoperator for measured R\'enyi relative entropy}
\begin{lem}
\label{lem:hessian-super-op-meas-renyi-rel-ent}For $\alpha\in\left(0,\frac{1}{2}\right)$,
the Hessian superoperator of $g_{\rho,\sigma}^{\alpha}(\omega)$ in
\eqref{eq:g-alpha-meas-renyi} is the following self-adjoint, Hermiticity-preserving
superoperator:
\begin{equation}
\Psi_{\alpha,\omega}(X)\coloneqq-c_{1}(\alpha)\int_{0}^{\infty}dt\ t^{\frac{\alpha}{\alpha-1}}\left(\omega+tI\right)^{-1}\left[X\left(\omega+tI\right)^{-1}\sigma+\sigma\left(\omega+tI\right)^{-1}X\right]\left(\omega+tI\right)^{-1},\label{eq:hessian-superoperator-alpha-0-half}
\end{equation}
where $c_{1}(\alpha)$ is defined in \eqref{eq:c1-alpha-constant}.

For $\alpha\in\left[\frac{1}{2},1\right)\cup\left(1,\infty\right)$,
the Hessian superoperator of $h_{\rho,\sigma}^{\alpha}(\omega)$ in
\eqref{eq:h-alpha-meas-renyi} is the following self-adjoint, Hermiticity-preserving
superoperator:
\begin{equation}
\Phi_{\alpha,\omega}(X)\coloneqq-c_{2}(\alpha)\int_{0}^{\infty}dt\ t^{\frac{\alpha-1}{\alpha}}\left(\omega+tI\right)^{-1}\left[X\left(\omega+tI\right)^{-1}\rho+\rho\left(\omega+tI\right)^{-1}X\right]\left(\omega+tI\right)^{-1},\label{eq:hessian-superoperator-alpha-half-infty}
\end{equation}
where $c_{2}(\alpha)$ is defined in \eqref{eq:c2-alpha-constant}.
\end{lem}

\begin{proof}
By recalling \eqref{eq:matrix-gradient-g}, consider that the Hessian
superoperator $\Psi_{\alpha,\omega}$ can be computed from
\begin{align}
\frac{\partial}{\partial\omega_{i,j}}\frac{\partial}{\partial\omega}g_{\rho,\sigma}^{\alpha}(\omega) & =\frac{\partial}{\partial\omega_{i,j}}\left(\alpha\rho+c_{1}(\alpha)\int_{0}^{\infty}dt\ t^{\frac{\alpha}{\alpha-1}}\left(\omega+tI\right)^{-1}\sigma\left(\omega+tI\right)^{-1}\right)\\
 & =c_{1}(\alpha)\int_{0}^{\infty}dt\ t^{\frac{\alpha}{\alpha-1}}\frac{\partial}{\partial\omega_{i,j}}\left[\left(\omega+tI\right)^{-1}\sigma\left(\omega+tI\right)^{-1}\right].
\end{align}
Now applying \eqref{eq:hessian-calc-1}--\eqref{eq:hessian-calc-final},
but with $\rho$ therein substituted with $\sigma$, we conclude \eqref{eq:hessian-superoperator-alpha-0-half}. 

The superoperator in \eqref{eq:hessian-superoperator-alpha-0-half}
is Hermiticity preserving by inspection.

Observing that
\begin{equation}
\Psi_{\alpha,\omega}=-c_{1}(\alpha)\int_{0}^{\infty}dt\ t^{\frac{\alpha}{\alpha-1}}\Phi_{t,\sigma,\omega},
\end{equation}
where $\Phi_{t,\sigma,\omega}$ is defined from \eqref{eq:Phi-s-rho-superop},
consider that $\Psi_{\alpha,\omega}$ is self-adjoint because $\Phi_{t,\sigma,\omega}$
is (recall \eqref{eq:self-adjoint-hessian-superop-1}--\eqref{eq:self-adjoint-hessian-superop-last});
i.e.,
\begin{align}
\left\langle Y,\Psi_{\alpha,\omega}(X)\right\rangle  & =-c_{1}(\alpha)\int_{0}^{\infty}dt\ t^{\frac{\alpha}{\alpha-1}}\left\langle Y,\Phi_{t,\sigma,\omega}(X)\right\rangle \\
 & =-c_{1}(\alpha)\int_{0}^{\infty}dt\ t^{\frac{\alpha}{\alpha-1}}\left\langle \Phi_{t,\sigma,\omega}(Y),X\right\rangle \\
 & =\left\langle \Psi_{\alpha,\omega}(Y),X\right\rangle .
\end{align}

By recalling \eqref{eq:matrix-gradient-h}, consider that the Hessian
superoperator $\Phi_{\alpha,\omega}$ can be computed from
\begin{align}
\frac{\partial}{\partial\omega_{i,j}}\frac{\partial}{\partial\omega}h_{\rho,\sigma}^{\alpha}(\omega) & =\frac{\partial}{\partial\omega_{i,j}}\left(\left(1-\alpha\right)\sigma+c_{2}(\alpha)\int_{0}^{\infty}dt\ t^{\frac{\alpha-1}{\alpha}}\left(\omega+tI\right)^{-1}\rho\left(\omega+tI\right)^{-1}\right)\\
 & =c_{2}(\alpha)\int_{0}^{\infty}dt\ t^{\frac{\alpha}{\alpha-1}}\frac{\partial}{\partial\omega_{i,j}}\left[\left(\omega+tI\right)^{-1}\rho\left(\omega+tI\right)^{-1}\right].
\end{align}
Now applying \eqref{eq:hessian-calc-1}--\eqref{eq:hessian-calc-final},
we conclude \eqref{eq:hessian-superoperator-alpha-0-half}. 

The superoperator in \eqref{eq:hessian-superoperator-alpha-half-infty}
is Hermiticity preserving by inspection.

Observing that
\begin{equation}
\Phi_{\alpha,\omega}=-c_{2}(\alpha)\int_{0}^{\infty}dt\ t^{\frac{\alpha-1}{\alpha}}\Phi_{t,\rho,\omega},
\end{equation}
consider that $\Phi_{\alpha,\omega}$ is self-adjoint because $\Phi_{t,\rho,\omega}$
is (recall \eqref{eq:self-adjoint-hessian-superop-1}--\eqref{eq:self-adjoint-hessian-superop-last});
i.e.,
\begin{align}
\left\langle Y,\Phi_{\alpha,\omega}(X)\right\rangle  & =-c_{2}(\alpha)\int_{0}^{\infty}dt\ t^{\frac{\alpha-1}{\alpha}}\left\langle Y,\Phi_{t,\rho,\omega}(X)\right\rangle \\
 & =-c_{2}(\alpha)\int_{0}^{\infty}dt\ t^{\frac{\alpha-1}{\alpha}}\left\langle \Phi_{t,\rho,\omega}(Y),X\right\rangle \\
 & =\left\langle \Phi_{\alpha,\omega}(Y),X\right\rangle ,
\end{align}
thus concluding the proof.
\end{proof}

\subsection{Smoothness and strong convexity / concavity for measured R\'enyi
relative entropy}
\begin{lem}
\label{lem:smoothness-strong-con-meas-renyi}For $\alpha\in\left(0,\frac{1}{2}\right)$
and $\omega>0$ satisfying \eqref{eq:optimal-omega-alpha-0-half},
the Hessian superoperator $\Psi_{\alpha,\omega}$ in \eqref{eq:hessian-superoperator-alpha-0-half}
has its minimum and maximum eigenvalues bounded as follows:
\begin{align}
\left(\frac{\alpha}{1-\alpha}\right)\lambda_{\min}(\sigma)\left\Vert \rho^{-\frac{1}{2}}\sigma\rho^{-\frac{1}{2}}\right\Vert ^{-\left(2-\alpha\right)} & \leq\min_{Y:\left\Vert Y\right\Vert _{2}=1}\left\langle Y,\Psi_{\alpha,\omega}(Y)\right\rangle \\
 & \leq\max_{Y:\left\Vert Y\right\Vert _{2}=1}\left\langle Y,\Psi_{\alpha,\omega}(Y)\right\rangle \\
 & \leq\left(\frac{\alpha}{1-\alpha}\right)\lambda_{\max}(\sigma)\left\Vert \sigma^{-\frac{1}{2}}\rho\sigma^{-\frac{1}{2}}\right\Vert ^{2-\alpha}.
\end{align}

For $\alpha\in\left[\frac{1}{2},1\right)$ and $\omega>0$ satisfying
\eqref{eq:optimal-omega-alpha-half-infty}, the Hessian superoperator
$\Phi_{\alpha,\omega}$ in \eqref{eq:hessian-superoperator-alpha-half-infty}
has its minimum and maximum eigenvalues bounded as follows:
\begin{align}
\lambda_{\min}(\rho)\left(\frac{1-\alpha}{\alpha}\right)\left\Vert \sigma^{-\frac{1}{2}}\rho\sigma^{-\frac{1}{2}}\right\Vert ^{-\left(\alpha+1\right)} & \leq\min_{Y:\left\Vert Y\right\Vert _{2}=1}\left\langle Y,\Phi_{\alpha,\omega}(Y)\right\rangle \\
 & \leq\max_{Y:\left\Vert Y\right\Vert _{2}=1}\left\langle Y,\Phi_{\alpha,\omega}(Y)\right\rangle \\
 & \leq\lambda_{\max}(\rho)\left(\frac{1-\alpha}{\alpha}\right)\left\Vert \rho^{-\frac{1}{2}}\sigma\rho^{-\frac{1}{2}}\right\Vert ^{\alpha+1}.
\end{align}

For $\alpha\in\left(1,\infty\right)$ and $\omega>0$ satisfying \eqref{eq:optimal-omega-alpha-half-infty},
the Hessian superoperator $\Phi_{\alpha,\omega}$ in \eqref{eq:hessian-superoperator-alpha-half-infty}
has its minimum and maximum eigenvalues bounded as follows:
\begin{align}
-\lambda_{\max}(\rho)\left(\frac{\alpha-1}{\alpha}\right)\left\Vert \rho^{-\frac{1}{2}}\sigma\rho^{-\frac{1}{2}}\right\Vert ^{\alpha+1} & \leq\min_{Y:\left\Vert Y\right\Vert _{2}=1}\left\langle Y,\Phi_{\alpha,\omega}(Y)\right\rangle \\
 & \leq\max_{Y:\left\Vert Y\right\Vert _{2}=1}\left\langle Y,\Phi_{\alpha,\omega}(Y)\right\rangle \\
 & \leq-\lambda_{\min}(\rho)\left(\frac{\alpha-1}{\alpha}\right)\left\Vert \sigma^{-\frac{1}{2}}\rho\sigma^{-\frac{1}{2}}\right\Vert ^{-\left(\alpha+1\right)}.
\end{align}
\end{lem}

\begin{proof}
For $\alpha\in\left(0,\frac{1}{2}\right)$, observe from \eqref{eq:hessian-superoperator-alpha-0-half}
that
\begin{equation}
\Psi_{\alpha,\omega}=-c_{1}(\alpha)\int_{0}^{\infty}dt\ t^{\frac{\alpha}{\alpha-1}}\Phi_{t,\sigma,\omega},
\end{equation}
where $\Phi_{t,\sigma,\omega}$ is defined in \eqref{eq:Phi-s-rho-superop}.
Observe that
\begin{equation}
-c_{1}(\alpha)=-\left(1-\alpha\right)\frac{\sin\!\left(\left(\frac{\alpha}{\alpha-1}\right)\pi\right)}{\pi}\geq0
\end{equation}
for $\alpha\in\left(0,\frac{1}{2}\right)$, and define
\begin{align}
u_{1}(\alpha) & \equiv\left\Vert \rho^{-\frac{1}{2}}\sigma\rho^{-\frac{1}{2}}\right\Vert ^{1-\alpha},\\
\ell_{1}(\alpha) & \equiv\left\Vert \sigma^{-\frac{1}{2}}\rho\sigma^{-\frac{1}{2}}\right\Vert ^{-\left(1-\alpha\right)}.
\end{align}
Applying \eqref{eq:Phi-s-rho-min-eig-1}--\eqref{eq:Phi-s-rho-min-eig-last},
but with $\rho$ therein substituted with $\sigma$, and considering
\eqref{eq:optimal-omega-alpha-0-half}, observe that
\begin{align}
\left\langle Y,\Psi_{\alpha,\omega}(Y)\right\rangle  & =-c_{1}(\alpha)\int_{0}^{\infty}dt\ t^{\frac{\alpha}{\alpha-1}}\left\langle Y,\Phi_{t,\sigma,\omega}(Y)\right\rangle \\
 & \geq-2c_{1}(\alpha)\lambda_{\min}(\sigma)\int_{0}^{\infty}dt\ t^{\frac{\alpha}{\alpha-1}}\left(u_{1}(\alpha)+t\right)^{-3}\\
 & =2\lambda_{\min}(\sigma)\frac{\alpha}{2\left(1-\alpha\right)}u_{1}(\alpha)^{-\frac{2-\alpha}{1-\alpha}}\\
 & =\left(\frac{\alpha}{1-\alpha}\right)\lambda_{\min}(\sigma)\left\Vert \rho^{-\frac{1}{2}}\sigma\rho^{-\frac{1}{2}}\right\Vert ^{\alpha-2},
\end{align}
and
\begin{align}
\left\langle Y,\Phi(Y)\right\rangle  & =-c_{1}(\alpha)\int_{0}^{\infty}dt\ t^{\frac{\alpha}{\alpha-1}}\left\langle Y,\Phi_{t,\sigma,\omega}(Y)\right\rangle \\
 & \leq-2c_{1}(\alpha)\lambda_{\max}(\sigma)\int_{0}^{\infty}dt\ t^{\frac{\alpha}{\alpha-1}}\left(\ell_{1}(\alpha)+t\right)^{-3}\\
 & =\left(\frac{\alpha}{1-\alpha}\right)\lambda_{\max}(\sigma)\left\Vert \sigma^{-\frac{1}{2}}\rho\sigma^{-\frac{1}{2}}\right\Vert ^{2-\alpha},
\end{align}
where, in both cases above, we used the following integral, holding
for $\alpha\in\left(0,\frac{1}{2}\right)$ and $k>0$:
\begin{equation}
-c_{1}(\alpha)\int_{0}^{\infty}dt\ t^{\frac{\alpha}{\alpha-1}}\left(k+t\right)^{-3}=\frac{\alpha}{2\left(1-\alpha\right)}k^{-\frac{2-\alpha}{1-\alpha}},\label{eq:key-integral-meas-renyi-eigenval-bnds-first}
\end{equation}
as derived in Appendix~\ref{app:Derivations-of-integral-formulas}.

For $\alpha\in\left[\frac{1}{2},1\right)$, observe from \eqref{eq:hessian-superoperator-alpha-half-infty}
that
\begin{equation}
\Phi_{\alpha,\omega}=-c_{2}(\alpha)\int_{0}^{\infty}dt\ t^{\frac{\alpha-1}{\alpha}}\Phi_{t,\rho,\omega},
\end{equation}
where $\Phi_{t,\rho,\omega}$ is defined in \eqref{eq:Phi-s-rho-superop}.
Observe that
\begin{equation}
-c_{2}(\alpha)=-\alpha\frac{\sin\!\left(\left(\frac{\alpha-1}{\alpha}\right)\pi\right)}{\pi}\geq0
\end{equation}
for $\alpha\in\left[\frac{1}{2},1\right)$, and define
\begin{align}
u_{2}(\alpha) & \equiv\left\Vert \sigma^{-\frac{1}{2}}\rho\sigma^{-\frac{1}{2}}\right\Vert ^{\alpha},\\
\ell_{2}(\alpha) & \equiv\left\Vert \rho^{-\frac{1}{2}}\sigma\rho^{-\frac{1}{2}}\right\Vert ^{-\alpha}.
\end{align}
Applying \eqref{eq:Phi-s-rho-min-eig-1}--\eqref{eq:Phi-s-rho-min-eig-last}
and considering \eqref{eq:optimal-omega-alpha-half-infty}, observe
that
\begin{align}
\left\langle Y,\Psi_{\alpha,\omega}(Y)\right\rangle  & =-c_{2}(\alpha)\int_{0}^{\infty}dt\ t^{\frac{\alpha-1}{\alpha}}\left\langle Y,\Phi_{t,\rho,\omega}(Y)\right\rangle \\
 & \geq-2c_{2}(\alpha)\lambda_{\min}(\rho)\int_{0}^{\infty}dt\ t^{\frac{\alpha-1}{\alpha}}\left(u_{2}(\alpha)+t\right)^{-3}\\
 & =2\lambda_{\min}(\rho)\left(\frac{1-\alpha}{2\alpha}\right)\left(\left\Vert \sigma^{-\frac{1}{2}}\rho\sigma^{-\frac{1}{2}}\right\Vert ^{\alpha}\right)^{-\frac{\alpha+1}{\alpha}}\\
 & =\lambda_{\min}(\rho)\left(\frac{1-\alpha}{\alpha}\right)\left\Vert \sigma^{-\frac{1}{2}}\rho\sigma^{-\frac{1}{2}}\right\Vert ^{-\left(\alpha+1\right)},
\end{align}
and
\begin{align}
\left\langle Y,\Phi(Y)\right\rangle  & =-c_{2}(\alpha)\int_{0}^{\infty}dt\ t^{\frac{\alpha-1}{\alpha}}\left\langle Y,\Phi_{t,\rho,\omega}(Y)\right\rangle \\
 & \leq-2c_{2}(\alpha)\lambda_{\max}(\rho)\int_{0}^{\infty}dt\ t^{\frac{\alpha-1}{\alpha}}\left(\ell_{2}(\alpha)+t\right)^{-3}\\
 & =2\lambda_{\max}(\rho)\left(\frac{1-\alpha}{2\alpha}\right)\left(\left\Vert \rho^{-\frac{1}{2}}\sigma\rho^{-\frac{1}{2}}\right\Vert ^{-\alpha}\right)^{-\frac{\alpha+1}{\alpha}}\\
 & =\lambda_{\max}(\rho)\left(\frac{1-\alpha}{\alpha}\right)\left\Vert \rho^{-\frac{1}{2}}\sigma\rho^{-\frac{1}{2}}\right\Vert ^{\alpha+1},
\end{align}
where, in both cases above, we used the following integral, holding
for $\alpha\in\left[\frac{1}{2},1\right)\cup(1,\infty)$ and $k>0$:
\begin{equation}
-c_{2}(\alpha)\int_{0}^{\infty}dt\ t^{\frac{\alpha-1}{\alpha}}\left(k+t\right)^{-3}=\left(\frac{1-\alpha}{2\alpha}\right)k^{-\frac{\alpha+1}{\alpha}},\label{eq:key-integral-meas-renyi-eigenval-bnds}
\end{equation}
as derived in Appendix~\ref{app:Derivations-of-integral-formulas}.

For $\alpha\in\left(1,\infty\right)$, observe from \eqref{eq:hessian-superoperator-alpha-half-infty}
that
\begin{equation}
\Phi_{\alpha,\omega}=-c_{2}(\alpha)\int_{0}^{\infty}dt\ t^{\frac{\alpha-1}{\alpha}}\Phi_{t,\rho,\omega},
\end{equation}
where $\Phi_{t,\rho,\omega}$ is defined in \eqref{eq:Phi-s-rho-superop}.
Observe that
\begin{equation}
c_{2}(\alpha)=\alpha\frac{\sin\!\left(\left(\frac{\alpha-1}{\alpha}\right)\pi\right)}{\pi}\geq0
\end{equation}
for $\alpha\in\left(1,\infty\right)$, and define
\begin{align}
u_{2}(\alpha) & \equiv\left\Vert \sigma^{-\frac{1}{2}}\rho\sigma^{-\frac{1}{2}}\right\Vert ^{\alpha},\\
\ell_{2}(\alpha) & \equiv\left\Vert \rho^{-\frac{1}{2}}\sigma\rho^{-\frac{1}{2}}\right\Vert ^{-\alpha}.
\end{align}
Applying \eqref{eq:Phi-s-rho-min-eig-1}--\eqref{eq:Phi-s-rho-min-eig-last}
and considering \eqref{eq:optimal-omega-alpha-half-infty}, observe
that
\begin{align}
\left\langle Y,\Psi_{\alpha,\omega}(Y)\right\rangle  & =-c_{2}(\alpha)\int_{0}^{\infty}dt\ t^{\frac{\alpha-1}{\alpha}}\left\langle Y,\Phi_{t,\rho,\omega}(Y)\right\rangle \\
 & \geq-2c_{2}(\alpha)\lambda_{\max}(\rho)\int_{0}^{\infty}dt\ t^{\frac{\alpha-1}{\alpha}}\left(\ell_{2}(\alpha)+t\right)^{-3}\\
 & =2\lambda_{\max}(\rho)\left(\frac{1-\alpha}{2\alpha}\right)\left(\left\Vert \rho^{-\frac{1}{2}}\sigma\rho^{-\frac{1}{2}}\right\Vert ^{-\alpha}\right)^{-\frac{\alpha+1}{\alpha}}\\
 & =-\lambda_{\max}(\rho)\left(\frac{\alpha-1}{\alpha}\right)\left\Vert \rho^{-\frac{1}{2}}\sigma\rho^{-\frac{1}{2}}\right\Vert ^{\alpha+1},
\end{align}
and
\begin{align}
\left\langle Y,\Phi(Y)\right\rangle  & =-c_{2}(\alpha)\int_{0}^{\infty}dt\ t^{\frac{\alpha-1}{\alpha}}\left\langle Y,\Phi_{t,\rho,\omega}(Y)\right\rangle \\
 & \leq-2c_{2}(\alpha)\lambda_{\min}(\rho)\int_{0}^{\infty}dt\ t^{\frac{\alpha-1}{\alpha}}\left(u_{2}(\alpha)+t\right)^{-3}\\
 & =2\lambda_{\min}(\rho)\left(\frac{1-\alpha}{2\alpha}\right)\left(\left\Vert \sigma^{-\frac{1}{2}}\rho\sigma^{-\frac{1}{2}}\right\Vert ^{\alpha}\right)^{-\frac{\alpha+1}{\alpha}}\\
 & =-\lambda_{\min}(\rho)\left(\frac{\alpha-1}{\alpha}\right)\left\Vert \sigma^{-\frac{1}{2}}\rho\sigma^{-\frac{1}{2}}\right\Vert ^{-\left(\alpha+1\right)},
\end{align}
where again we applied \eqref{eq:key-integral-meas-renyi-eigenval-bnds}.
\end{proof}
\begin{cor}
\label{cor:smoothness-strong-concavity-meas-renyi-rel-ent}For $\alpha\in\left(0,\frac{1}{2}\right)$
and on the operator interval
\begin{equation}
\omega\in\left[\left\Vert \sigma^{-\frac{1}{2}}\rho\sigma^{-\frac{1}{2}}\right\Vert ^{-\left(1-\alpha\right)}I,\left\Vert \rho^{-\frac{1}{2}}\sigma\rho^{-\frac{1}{2}}\right\Vert ^{1-\alpha}I\right],\label{eq:a-0-half-op-interval}
\end{equation}
the function $g_{\rho,\sigma}^{\alpha}(\omega)$ is $\beta_{\alpha}$-smooth
with parameter
\begin{equation}
\beta_{\alpha}\coloneqq\left(\frac{\alpha}{1-\alpha}\right)\lambda_{\max}(\sigma)\left\Vert \sigma^{-\frac{1}{2}}\rho\sigma^{-\frac{1}{2}}\right\Vert ^{2-\alpha},
\end{equation}
and it is $\gamma_{\alpha}$-strongly convex with parameter
\begin{equation}
\gamma_{\alpha}\coloneqq\left(\frac{\alpha}{1-\alpha}\right)\lambda_{\min}(\sigma)\left\Vert \rho^{-\frac{1}{2}}\sigma\rho^{-\frac{1}{2}}\right\Vert ^{\alpha-2}.
\end{equation}
Thus, on the operator interval in \eqref{eq:a-0-half-op-interval}
and for $\alpha\in\left(0,\frac{1}{2}\right)$, the condition number
$\kappa_{\alpha}$ of the function $g_{\rho,\sigma}^{\alpha}(\omega)$
is given by
\begin{equation}
\kappa_{\alpha}\coloneqq\frac{\beta_{\alpha}}{\gamma_{\alpha}}=\frac{\lambda_{\max}(\sigma)}{\lambda_{\min}(\sigma)}\left(\left\Vert \rho^{-\frac{1}{2}}\sigma\rho^{-\frac{1}{2}}\right\Vert \left\Vert \sigma^{-\frac{1}{2}}\rho\sigma^{-\frac{1}{2}}\right\Vert \right)^{2-\alpha}.\label{eq:kappa-alpha-0-half}
\end{equation}

For $\alpha\in\left[\frac{1}{2},1\right)\cup\left(1,\infty\right)$
and on the operator interval
\begin{equation}
\omega\in\left[\left\Vert \rho^{-\frac{1}{2}}\sigma\rho^{-\frac{1}{2}}\right\Vert ^{-\alpha}I,\left\Vert \sigma^{-\frac{1}{2}}\rho\sigma^{-\frac{1}{2}}\right\Vert ^{\alpha}I\right],\label{eq:a-half-infty-op-interval}
\end{equation}
the function $h_{\rho,\sigma}^{\alpha}(\omega)$ is $\beta_{\alpha}$-smooth
with parameter
\begin{equation}
\beta_{\alpha}\coloneqq\lambda_{\max}(\rho)\left(\frac{1-\alpha}{\alpha}\right)\left\Vert \rho^{-\frac{1}{2}}\sigma\rho^{-\frac{1}{2}}\right\Vert ^{\alpha+1},
\end{equation}
and it is $\gamma_{\alpha}$-strongly convex / concave with parameter
\begin{equation}
\gamma_{\alpha}\coloneqq\lambda_{\min}(\rho)\left(\frac{1-\alpha}{\alpha}\right)\left\Vert \sigma^{-\frac{1}{2}}\rho\sigma^{-\frac{1}{2}}\right\Vert ^{-\left(\alpha+1\right)}.
\end{equation}
Thus, on the operator interval in \eqref{eq:a-half-infty-op-interval}
and for $\alpha\in\left[\frac{1}{2},1\right)\cup\left(1,\infty\right)$,
the condition number $\kappa_{\alpha}$ of the function $h_{\rho,\sigma}^{\alpha}(\omega)$
is given by
\begin{equation}
\kappa_{\alpha}\coloneqq\frac{\beta_{\alpha}}{\gamma_{\alpha}}=\frac{\lambda_{\max}(\rho)}{\lambda_{\min}(\rho)}\left(\left\Vert \rho^{-\frac{1}{2}}\sigma\rho^{-\frac{1}{2}}\right\Vert \left\Vert \sigma^{-\frac{1}{2}}\rho\sigma^{-\frac{1}{2}}\right\Vert \right)^{\alpha+1}.\label{eq:kappa-alpha-half-infty}
\end{equation}
\end{cor}

\begin{proof}
This is an immediate consequence of the eigenvalue bounds from Lemma
\ref{lem:smoothness-strong-con-meas-renyi}.
\end{proof}

\subsection{Nesterov accelerated optimization of measured R\'enyi relative entropy}

This section presents our algorithms for calculating the measured
R\'enyi relative quasi-entropy, as given in \eqref{eq:meas-renyi-var-form}.
We consider the two cases of $\alpha\in\left(0,\frac{1}{2}\right)$
and $\alpha\in\left[\frac{1}{2},1\right)\cup\left(1,\infty\right)$
separately.

\subsubsection{Optimization algorithm for $\alpha\in\left(0,\frac{1}{2}\right)$}

We employ Nesterov accelerated projected gradient descent in order
to perform the optimization in \eqref{eq:meas-renyi-var-form} for
$\alpha\in\left(0,\frac{1}{2}\right)$.
\begin{lyxalgorithm}
\label{alg:measured-renyi-0-half}For $\alpha\in\left(0,\frac{1}{2}\right)$,
proceed according to the following steps:
\begin{enumerate}
\item Set $m=0$, and initialize
\begin{equation}
\omega_{0}\leftarrow\xi_{0}\leftarrow\frac{1}{2}\left(\left\Vert \sigma^{-\frac{1}{2}}\rho\sigma^{-\frac{1}{2}}\right\Vert ^{-\left(1-\alpha\right)}+\left\Vert \rho^{-\frac{1}{2}}\sigma\rho^{-\frac{1}{2}}\right\Vert ^{1-\alpha}\right)I.
\end{equation}
\item While $\left\Vert \left.\frac{\partial}{\partial\omega}g_{\rho,\sigma}^{\alpha}(\omega)\right|_{\omega=\omega_{m}}\right\Vert _{2}>\sqrt{2\gamma_{\alpha}\varepsilon}$,
\begin{align}
\xi_{m+1} & \leftarrow\omega_{m}-\frac{1}{\beta_{\alpha}}\left.\frac{\partial}{\partial\omega}g_{\rho,\sigma}^{\alpha}(\omega)\right|_{\omega=\omega_{m}},\\
\chi_{m+1} & \leftarrow\left[\xi_{m+1}\right]_{\mathcal{I}_{\alpha}},\\
\omega_{m+1} & \leftarrow\chi_{m+1}+\left(\frac{\sqrt{\kappa_{\alpha}}-1}{\sqrt{\kappa_{\alpha}}+1}\right)\left(\chi_{m+1}-\chi_{m}\right),
\end{align}
where
\begin{equation}
\mathcal{I}_{\alpha}\coloneqq\left[\left\Vert \sigma^{-\frac{1}{2}}\rho\sigma^{-\frac{1}{2}}\right\Vert ^{-\left(1-\alpha\right)},\left\Vert \rho^{-\frac{1}{2}}\sigma\rho^{-\frac{1}{2}}\right\Vert ^{1-\alpha}\right],
\end{equation}
and $\beta_{\alpha}$, $\gamma_{\alpha}$, and $\kappa_{\alpha}$
are defined in Corollary~\ref{cor:smoothness-strong-concavity-meas-renyi-rel-ent}.
Set $m=m+1$.
\item Output $g_{\rho,\sigma}^{\alpha}(\omega_{m})$ as an estimate of $Q_{\alpha}^{M}(\rho\|\sigma)$
satisfying \eqref{eq:meas-Renyi-rel-ent-eps-error-alpha-0-half}.
\end{enumerate}
\end{lyxalgorithm}

The stopping condition $\left\Vert \left.\frac{\partial}{\partial\omega}g_{\rho,\sigma}^{\alpha}(\omega)\right|_{\omega=\omega_{m}}\right\Vert _{2}\leq\sqrt{2\gamma_{\alpha}\varepsilon}$
implies that 
\begin{equation}
\left|Q_{\alpha}^{M}(\rho\|\sigma)-g_{\rho,\sigma}^{\alpha}(\omega_{m})\right|\leq\varepsilon,\label{eq:meas-Renyi-rel-ent-eps-error-alpha-0-half}
\end{equation}
as a consequence of the same reasoning that led to \eqref{eq:error-threshold-desired}.
The number of iterations required by the above algorithm is equal
to \cite[Theorem~3.18]{Bubeck2015}
\begin{equation}
O\!\left(\sqrt{\kappa_{\alpha}}\ln\!\left(\frac{1}{\varepsilon}\right)\right).
\end{equation}

\subsubsection{Optimization algorithm for $\alpha\in\left[\frac{1}{2},1\right)\cup\left(1,\infty\right)$}

We employ Nesterov accelerated gradient descent / ascent in order
to perform the optimization in \eqref{eq:meas-renyi-var-form} for
$\alpha\in\left[\frac{1}{2},1\right)\cup\left(1,\infty\right)$.
\begin{lyxalgorithm}
\label{alg:measured-renyi-half-infty}For $\alpha\in\left[\frac{1}{2},1\right)\cup\left(1,\infty\right)$,
proceed according to the following steps:
\begin{enumerate}
\item Set $m=0$, and initialize
\begin{equation}
\omega_{0}\leftarrow\xi_{0}\leftarrow\frac{1}{2}\left(\left\Vert \rho^{-\frac{1}{2}}\sigma\rho^{-\frac{1}{2}}\right\Vert ^{-\alpha}+\left\Vert \sigma^{-\frac{1}{2}}\rho\sigma^{-\frac{1}{2}}\right\Vert ^{\alpha}\right)I.
\end{equation}
\item While $\left\Vert \left.\frac{\partial}{\partial\omega}h_{\rho,\sigma}^{\alpha}(\omega)\right|_{\omega=\omega_{m}}\right\Vert _{2}>\sqrt{2\gamma_{\alpha}\varepsilon}$,
\begin{align}
\xi_{m+1} & \leftarrow\omega_{m}+\sgn(\alpha-1)\frac{1}{\beta_\alpha}\left.\frac{\partial}{\partial\omega}h_{\rho,\sigma}^{\alpha}(\omega)\right|_{\omega=\omega_{m}},\\
\chi_{m+1} & \leftarrow\left[\xi_{m+1}\right]_{\mathcal{I}_{\alpha}},\\
\omega_{m+1} & \leftarrow\chi_{m+1}+\left(\frac{\sqrt{\kappa_{\alpha}}-1}{\sqrt{\kappa_{\alpha}}+1}\right)\left(\chi_{m+1}-\chi_{m}\right),
\end{align}
where
\begin{equation}
\mathcal{I}_{\alpha}\coloneqq\left[\left\Vert \rho^{-\frac{1}{2}}\sigma\rho^{-\frac{1}{2}}\right\Vert ^{-\alpha},\left\Vert \sigma^{-\frac{1}{2}}\rho\sigma^{-\frac{1}{2}}\right\Vert ^{\alpha}\right],
\end{equation}
and $\beta_{\alpha}$, $\gamma_{\alpha}$, and $\kappa_{\alpha}$
are defined in Corollary~\ref{cor:smoothness-strong-concavity-meas-renyi-rel-ent}.
Set $m=m+1$.
\item Output $h_{\rho,\sigma}^{\alpha}(\omega_{m})$ as an estimate of $Q_{\alpha}^{M}(\rho\|\sigma)$
satisfying \eqref{eq:meas-Renyi-rel-ent-eps-error-alpha-half-infty}.
\end{enumerate}
\end{lyxalgorithm}

The stopping condition $\left\Vert \left.\frac{\partial}{\partial\omega}h_{\rho,\sigma}^{\alpha}(\omega)\right|_{\omega=\omega_{m}}\right\Vert _{2}\leq\sqrt{2\gamma_{\alpha}\varepsilon}$
ensures that 
\begin{equation}
\left|Q_{\alpha}^{M}(\rho\|\sigma)-h_{\rho,\sigma}^{\alpha}(\omega_{m})\right|\leq\varepsilon,\label{eq:meas-Renyi-rel-ent-eps-error-alpha-half-infty}
\end{equation}
as a consequence of the same reasoning that led to \eqref{eq:error-threshold-desired}.
The number of iterations required by the above algorithm is equal
to \cite[Theorem~3.18]{Bubeck2015}
\begin{equation}
O\!\left(\sqrt{\kappa_{\alpha}}\ln\!\left(\frac{1}{\varepsilon}\right)\right).
\end{equation}

\section{Comparison with SDP algorithms from~\cite{Huang2024}}

\label{sec:Comparison-with-SDP}

In our prior work~\cite{Huang2024},
we proposed semidefinite optimization algorithms for calculating
measured relative entropy and measured R\'enyi relative entropy.
In this section, we compare our algorithms proposed here with our
previous algorithms from~\cite{Huang2024}. In short, the main message
is that our proposed algorithm here is always more memory efficient,
and for well-conditioned states, it is notably faster.

Let us first consider the total computational costs of our algorithms
proposed here. The iteration complexity (total number of iterations
needed for convergence) of Algorithm~\ref{alg:measured-rel-ent},
for calculating measured relative entropy, is given by
\begin{equation}
O\!\left(\sqrt{\kappa}\ln\!\left(\frac{1}{\varepsilon}\right)\right),
\end{equation}
where $\kappa=\frac{\lambda_{\max}(\rho)}{\lambda_{\min}(\rho)}\left(\left\Vert \rho^{-\frac{1}{2}}\sigma\rho^{-\frac{1}{2}}\right\Vert \left\Vert \sigma^{-\frac{1}{2}}\rho\sigma^{-\frac{1}{2}}\right\Vert \right)^{2}$,
as stated in \eqref{eq:kappa-meas-rel-ent}. Furthermore, each iteration
of Algorithm~\ref{alg:measured-rel-ent} requires matrix gradient
updates and computation of $\left[\xi_{m+1}\right]_{\mathcal{I}}$,
the latter of which requires an eigenvalue decomposition. Both of
these steps have cost $O(d^{3})$, where $d$ is the dimension of
the states $\rho$ and $\sigma$. Thus, the total computational cost
of Algorithm~\ref{alg:measured-rel-ent} is
\begin{equation}
O\!\left(\sqrt{\kappa}d^{3}\ln\!\left(\frac{1}{\varepsilon}\right)\right).\label{eq:cost-meas-rel-ent-alg}
\end{equation}
The total computational costs of Algorithms~\ref{alg:measured-renyi-0-half}
and~\ref{alg:measured-renyi-half-infty} follow from similar arguments
and are given by
\begin{equation}
O\!\left(\sqrt{\kappa_{\alpha}}d^{3}\ln\!\left(\frac{1}{\varepsilon}\right)\right),\label{eq:cost-meas-Renyi-rel-ent-alg}
\end{equation}
where
\begin{align}
\kappa_{\alpha} & =\frac{\lambda_{\max}(\sigma)}{\lambda_{\min}(\sigma)}\left(\left\Vert \rho^{-\frac{1}{2}}\sigma\rho^{-\frac{1}{2}}\right\Vert \left\Vert \sigma^{-\frac{1}{2}}\rho\sigma^{-\frac{1}{2}}\right\Vert \right)^{2-\alpha}\quad\text{for}\ \alpha\in\left(0,\frac{1}{2}\right),\\
\kappa_{\alpha} & =\frac{\lambda_{\max}(\rho)}{\lambda_{\min}(\rho)}\left(\left\Vert \rho^{-\frac{1}{2}}\sigma\rho^{-\frac{1}{2}}\right\Vert \left\Vert \sigma^{-\frac{1}{2}}\rho\sigma^{-\frac{1}{2}}\right\Vert \right)^{\alpha+1}\quad\text{for}\ \alpha\in\left[\frac{1}{2},1\right)\cup\left(1,\infty\right),
\end{align}
as stated in \eqref{eq:kappa-alpha-0-half} and \eqref{eq:kappa-alpha-half-infty},
respectively. The total memory requirements of Algorithms~\ref{alg:measured-rel-ent},
\ref{alg:measured-renyi-0-half}, and~\ref{alg:measured-renyi-half-infty}
is equal to $O\!\left(d^{2}\right)$, as only a constant number of
$d\times d$ matrices are needed to be stored for each gradient update
step.

Our previous algorithms from~\cite{Huang2024} rely on semidefinite
optimization, which in turn makes use of the interior-point method
\cite{Nesterov1994} (see also~\cite{Vandenberghe1996,Helmberg1996,Todd2001,Helmberg2003}).
Path-following interior-points methods for optimizing over the positive
semidefinite cone of dimension $d'$ have iteration complexity
\begin{equation}
O\!\left(\sqrt{d'}\ln\!\left(\frac{1}{\varepsilon}\right)\right).
\end{equation}
Each iteration of the interior-point method requires forming and solving
a Newton system, also known as a Schur complement system. If the SDP
has $m$ affine constraints, then the Schur complement is an $m\times m$
system. For dense data, forming the Schur matrix typically costs at
least $O(m^{2}d'{}^{2})$, factoring it costs $O(m^{3})$, and various
matrix manipulations cost $O(d'{}^{3})$. Thus, the total cost per
iteration for the interior-point method is
\begin{equation}
O(m^{2}d'{}^{2}+d'{}^{3}+m^{3}),
\end{equation}
so that the total computational cost of the interior-point method
is
\begin{equation}
O\!\left(\left(m^{2}d'{}^{\frac{5}{2}}+d'{}^{\frac{7}{2}}+m^{3}d'{}^{\frac{1}{2}}\right)\ln\!\left(\frac{1}{\varepsilon}\right)\right).\label{eq:IPM-total-cost}
\end{equation}

Our algorithm from~\cite{Huang2024} for calculating measured relative
entropy involves solving an SDP with $O\!\left(\sqrt{\ln\left(\frac{1}{\varepsilon}\right)}\right)$
matrix inequalities, each of dimension $2d$, where $d$ is the dimension
of the states $\rho$ and $\sigma$. By examining the error analysis
in~\cite{Fawzi2019} more closely, in particular the proof of Theorem
1 therein (see \cite[Appendix~B.2.2]{Fawzi2019}), we find that the
number of matrix inequalities required is given by
\begin{equation}
m=\left\lceil \log_{2}\ln a\right\rceil +1+\frac{3}{2}\left(\left\lceil \sqrt{\log_{2}\left(\frac{32\ln a}{\varepsilon}\right)}\right\rceil +1\right),
\end{equation}
where, in our case, $a\coloneqq\max\!\left\{ \left\Vert \rho^{-\frac{1}{2}}\sigma\rho^{-\frac{1}{2}}\right\Vert ,\left\Vert \sigma^{-\frac{1}{2}}\rho\sigma^{-\frac{1}{2}}\right\Vert \right\} $,
due to Lemma~\ref{lem:omega-bounds}. Thus, by plugging into \eqref{eq:IPM-total-cost}
and ignoring the term $\log_{2}\ln a$, the total computational cost
of calculating measured relative entropy in this way is given by
\begin{equation}
O\!\left(\left[\ln\!\left(\frac{1}{\varepsilon}\right)\right]^{2}d'{}^{\frac{5}{2}}+d'{}^{\frac{7}{2}}\ln\!\left(\frac{1}{\varepsilon}\right)+\ln\!\left(\frac{1}{\varepsilon}\right)^{\frac{5}{2}}d'{}^{\frac{1}{2}}\right),
\end{equation}
where $d'=2d$. Comparing this cost with the cost in \eqref{eq:cost-meas-rel-ent-alg},
the biggest gain occurs if $\kappa\ll2d$, which is the case for well-conditioned
states.

Our algorithm from~\cite{Huang2024} for calculating measured R\'enyi
relative entropy involves solving an SDP with $O\!\left(\ln\!\left(q\right)\right)$
matrix inequalities, each of dimension $2d$, where $d$ is the dimension
of the states $\rho$ and $\sigma$. Here, we assume that the R\'enyi
parameter $\alpha$ is rational and can be written as $\frac{p}{q}=\frac{\alpha}{\alpha-1}$
for $\alpha\in\left(0,\frac{1}{2}\right)$ and as $\frac{p}{q}=1-\frac{1}{\alpha}$
for $\alpha\in\left[\frac{1}{2},1\right)\cup\left(1,\infty\right)$,
where $p$ and $q$ are relatively prime. Thus, by plugging into \eqref{eq:IPM-total-cost},
the total computational cost of calculating measured R\'enyi relative
entropy in this way is given by
\begin{equation}
O\!\left(\left(\ln\!\left(q\right)^{2}d'{}^{\frac{5}{2}}+d'{}^{\frac{7}{2}}+\ln\!\left(q\right)^{3}d'{}^{\frac{1}{2}}\right)\ln\!\left(\frac{1}{\varepsilon}\right)\right),
\end{equation}
where $d'=2d$. Comparing this with the cost in \eqref{eq:cost-meas-Renyi-rel-ent-alg},
the biggest gain occurs if $\kappa_{\alpha}\ll2d$, which is the case
for well-conditioned states.

In Figure~\ref{fig:run_time_linear_linear}, we compare the run time of the semidefinite program in~\cite{Huang2024} to that of Algorithm~\ref{alg:measured-rel-ent} as a function of the Hilbert space dimension. Here we chose $\varepsilon = \exp(-6)\approx 2.5\times 10^{-3}$. The input states are bosonic thermal states~\cite{Serafini17}, truncated to dimension $d$ in the Fock basis and then trace-normalized, with mean photon numbers $\tfrac{1}{10}d$ and $\tfrac{1.2}{10}d$, respectively. 
Across the entire range of dimensions shown, Algorithm~\ref{alg:measured-rel-ent} is consistently much faster than the SDP from~\cite{Huang2024}.  For example, around $d\approx 10$, the new algorithm is roughly two orders of magnitude times faster, while for $d\approx 80$, it is three orders of magnitude times faster. All python source files needed to reproduce these figures are available with the arXiv posting of our paper.

\begin{figure}[h!]
\includegraphics[width=0.48\linewidth]{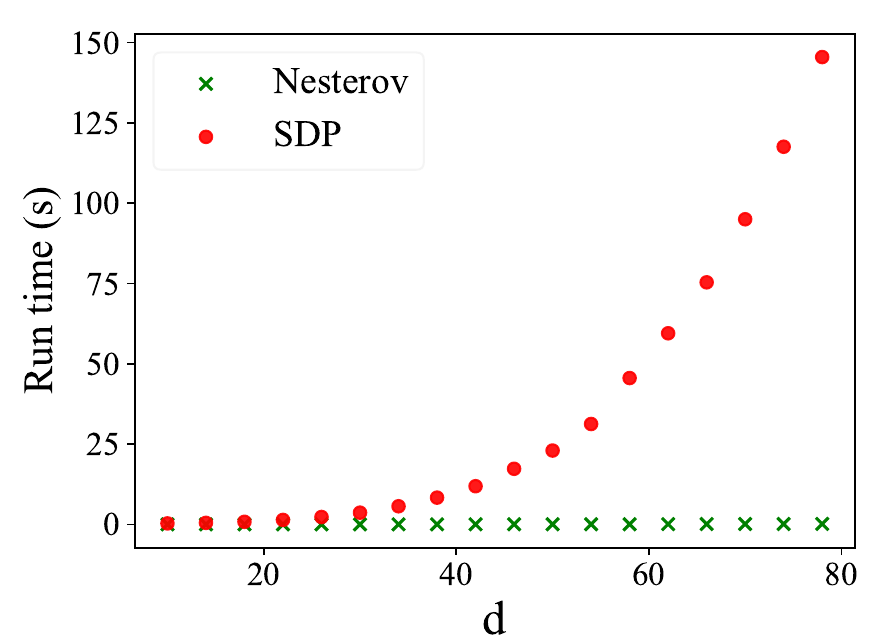}
\includegraphics[width=0.48\linewidth]{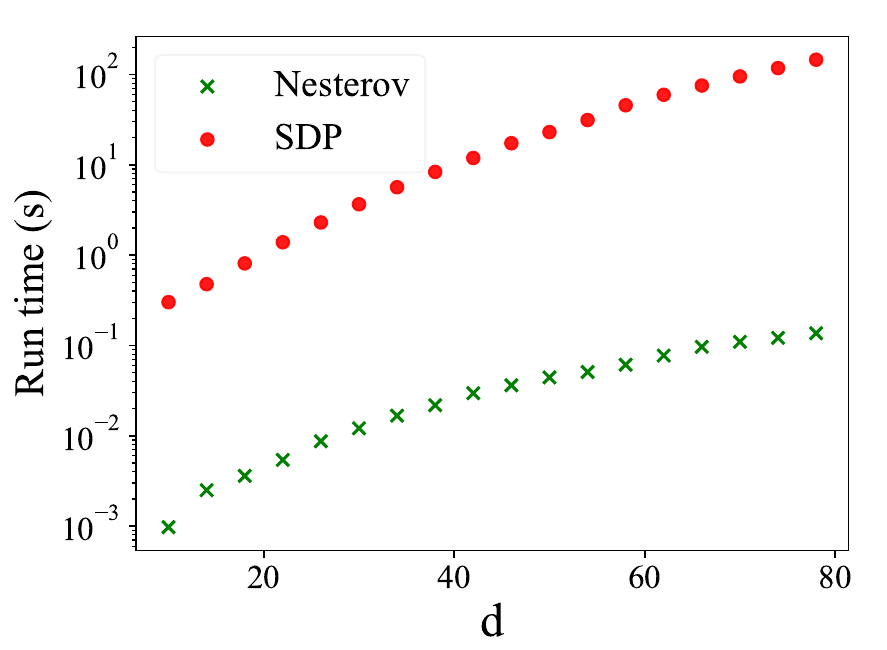}
\caption{\label{fig:run_time_linear_linear}  Runtime comparison between Algorithm~\ref{alg:measured-rel-ent} (green crosses) and the semidefinite program (SDP; red circles) of Ref.~\cite{Huang2024}. The left panel shows the data on linear axes, while the right panel uses a log-linear scale.}
\end{figure}

\section{Conclusion}

\label{sec:Conclusion}In summary, our main mathematical contribution
was to elucidate foundational properties of variational formulas for
measured relative entropies, namely, that their objective functions
are $\beta$-smooth and $\gamma$-strongly convex / concave on relevant
operator intervals. This enabled the application of Nesterov accelerated
projected gradient descent / ascent for calculating the optimal values
of measured relative entropies. The approach is generally more memory
efficient than our previous algorithm from~\cite{Huang2024}, which
is based on semidefinite optimization, and for well-conditioned states,
it is notably faster.

Going forward from here, let us recall that one of the main contributions
of~\cite{Huang2024} was semidefinite optimization algorithms for
calculating measured relative entropies of quantum channels. Here
we have not addressed how to conduct this optimization using Nesterov
accelerated methods. However, we think it might be possible to combine
ideas from the present paper and our prior paper~\cite{Huang2024}
to address this problem, but we leave it for future work.

\medskip{}

\textit{Acknowledgements}---ZH is supported by an ARC DECRA Fellowship
(DE230100144) “Quantum-enabled super-resolution imaging” and an RMIT
Vice Chancellor’s Senior Research Fellowship. She is also grateful
to the Cornell School of Electrical and Computer Engineering for hospitality
during an October 2025 research visit, as well as the Cornell Lab
of Ornithology for the Merlin app, which enabled many hours of fruitful
birding. MMW acknowledges support from the National Science Foundation
under grant no.~2329662 and from the Cornell School of Electrical
and Computer Engineering.

\bibliographystyle{alphaurl}
\phantomsection\addcontentsline{toc}{section}{\refname}\bibliography{Ref}

\begin{thebibliography}{HRVW96}

\bibitem[BFT17]{Berta2017}
Mario Berta, Omar Fawzi, and Marco Tomamichel.
\newblock On variational expressions for quantum relative entropies.
\newblock {\em Letters in Mathematical Physics}, 107(12):2239--2265, September 2017.
\newblock URL: \url{http://dx.doi.org/10.1007/s11005-017-0990-7}, \href {https://arxiv.org/abs/1512.02615v2} {\path{arXiv:1512.02615v2}}, \href {https://doi.org/10.1007/s11005-017-0990-7} {\path{doi:10.1007/s11005-017-0990-7}}.

\bibitem[BT09]{Beck2009}
Amir Beck and Marc Teboulle.
\newblock A fast iterative shrinkage-thresholding algorithm for linear inverse problems.
\newblock {\em SIAM Journal on Imaging Sciences}, 2(1):183--202, 2009.
\newblock \href {https://doi.org/10.1137/080716542} {\path{doi:10.1137/080716542}}.

\bibitem[Bub15]{Bubeck2015}
S{\'e}bastien Bubeck.
\newblock Convex optimization: Algorithms and complexity.
\newblock {\em Foundations and Trends in Machine Learning}, 8(3-4):231--357, 2015.
\newblock \href {https://doi.org/10.1561/2200000050} {\path{doi:10.1561/2200000050}}.

\bibitem[Dat09]{Datta2009}
Nilanjana Datta.
\newblock Min- and max-relative entropies and a new entanglement monotone.
\newblock {\em IEEE Transactions on Information Theory}, 55(6):2816--2826, June 2009.
\newblock arXiv:0803.2770.
\newblock \href {https://doi.org/10.1109/TIT.2009.2018325} {\path{doi:10.1109/TIT.2009.2018325}}.

\bibitem[DB16]{cvxpy2018}
Steven Diamond and Stephen Boyd.
\newblock {CVXPY}: A {P}ython-embedded modeling language for convex optimization.
\newblock {\em Journal of Machine Learning Research}, 17(83):1--5, 2016.
\newblock URL: \url{https://dl.acm.org/doi/10.5555/2946645.3007036}, \href {https://arxiv.org/abs/1603.00943} {\path{arXiv:1603.00943}}.

\bibitem[Don86]{Donald1986}
Matthew~J. Donald.
\newblock On the relative entropy.
\newblock {\em Communications in Mathematical Physics}, 105(1):13--34, 1986.
\newblock \href {https://doi.org/10.1007/BF01212339} {\path{doi:10.1007/BF01212339}}.

\bibitem[FS17]{Fawzi2017}
Hamza Fawzi and James Saunderson.
\newblock Lieb's concavity theorem, matrix geometric means, and semidefinite optimization.
\newblock {\em Linear Algebra and its Applications}, 513:240--263, 2017.
\newblock arXiv:1512.03401.
\newblock URL: \url{https://www.sciencedirect.com/science/article/pii/S0024379516304852}, \href {https://doi.org/10.1016/j.laa.2016.10.012} {\path{doi:10.1016/j.laa.2016.10.012}}.

\bibitem[FSP19]{Fawzi2019}
Hamza Fawzi, James Saunderson, and Pablo~A. Parrilo.
\newblock Semidefinite approximations of the matrix logarithm.
\newblock {\em Foundations of Computational Mathematics}, 19(2):259--296, 2019.
\newblock Package cvxquad at \url{https://github.com/hfawzi/cvxquad}.
\newblock \href {https://doi.org/10.1007/s10208-018-9385-0} {\path{doi:10.1007/s10208-018-9385-0}}.

\bibitem[Fuc96]{Fuchs1996}
Christopher Fuchs.
\newblock {\em Distinguishability and Accessible Information in Quantum Theory}.
\newblock PhD thesis, University of New Mexico, December 1996.
\newblock \href {https://arxiv.org/abs/quant-ph/9601020} {\path{arXiv:quant-ph/9601020}}.

\bibitem[GB14]{cvx2014}
Michael Grant and Stephen Boyd.
\newblock {CVX}: Matlab software for disciplined convex programming.
\newblock \url{http://cvxr.com/cvx}, 2014.

\bibitem[GG24]{Garrigos2024}
Guillaume Garrigos and Robert~M. Gower.
\newblock Handbook of convergence theorems for (stochastic) gradient methods, 2024.
\newblock URL: \url{https://arxiv.org/abs/2301.11235}, \href {https://arxiv.org/abs/2301.11235} {\path{arXiv:2301.11235}}.

\bibitem[GPSW24]{Goldfeld2024}
Ziv Goldfeld, Dhrumil Patel, Sreejith Sreekumar, and Mark~M. Wilde.
\newblock Quantum neural estimation of entropies.
\newblock {\em Physical Review A}, 109:032431, March 2024.
\newblock \href {https://doi.org/10.1103/PhysRevA.109.032431} {\path{doi:10.1103/PhysRevA.109.032431}}.

\bibitem[Hel02]{Helmberg2003}
Christoph Helmberg.
\newblock Semidefinite programming.
\newblock {\em European Journal of Operational Research}, 137(3):461--482, 2002.
\newblock URL: \url{https://www.sciencedirect.com/science/article/pii/S0377221701001436}, \href {https://doi.org/10.1016/S0377-2217(01)00143-6} {\path{doi:10.1016/S0377-2217(01)00143-6}}.

\bibitem[Hia21]{Hiai2021}
Fumio Hiai.
\newblock {\em Quantum $f$-divergences in von {N}eumann Algebras: Reversibility of Quantum Operations}.
\newblock Mathematical Physics Studies. Springer Singapore, 2021.
\newblock \href {https://doi.org/10.1007/978-981-33-4199-9} {\path{doi:10.1007/978-981-33-4199-9}}.

\bibitem[HM17]{Hiai2017}
Fumio Hiai and Milán Mosonyi.
\newblock Different quantum $f$-divergences and the reversibility of quantum operations.
\newblock {\em Reviews in Mathematical Physics}, 29(07):1750023, 2017.
\newblock \href {https://arxiv.org/abs/1604.03089} {\path{arXiv:1604.03089}}, \href {https://doi.org/10.1142/S0129055X17500234} {\path{doi:10.1142/S0129055X17500234}}.

\bibitem[HRVW96]{Helmberg1996}
Christoph Helmberg, Franz Rendl, Robert~J. Vanderbei, and Henry Wolkowicz.
\newblock An interior-point method for semidefinite programming.
\newblock {\em SIAM Journal on Optimization}, 6(2):342--361, 1996.
\newblock \href {https://doi.org/10.1137/0806020} {\path{doi:10.1137/0806020}}.

\bibitem[HW24]{Huang2024}
Zixin Huang and Mark~M. Wilde.
\newblock Semi-definite optimization of the measured relative entropies of quantum states and channels, 2024.
\newblock URL: \url{https://arxiv.org/abs/2406.19060}, \href {https://arxiv.org/abs/2406.19060} {\path{arXiv:2406.19060}}.

\bibitem[L{\"o}f04]{yalmip2022}
Johan L{\"o}fberg.
\newblock Yalmip: A toolbox for modeling and optimization in matlab.
\newblock In {\em Proceedings of the 2004 IEEE International Symposium on Computer Aided Control Systems Design}, pages 284--289, 2004.
\newblock \href {https://doi.org/10.1109/CACSD.2004.1393890} {\path{doi:10.1109/CACSD.2004.1393890}}.

\bibitem[Mat16]{Matsumoto2016}
Keiji Matsumoto.
\newblock On maximization of measured $f$-divergence between a given pair of quantum states, 2016.
\newblock URL: \url{https://arxiv.org/abs/1412.3676}, \href {https://arxiv.org/abs/1412.3676} {\path{arXiv:1412.3676}}.

\bibitem[MN19]{Magnus2019}
Jan~R. Magnus and Heinz Neudecker.
\newblock {\em Matrix Differential Calculus with Applications in Statistics and Econometrics}.
\newblock John Wiley \& Sons, Chichester, 3 edition, 2019.
\newblock \href {https://doi.org/10.1002/9781119541219} {\path{doi:10.1002/9781119541219}}.

\bibitem[MO15]{Mosonyi2015}
Milán Mosonyi and Tomohiro Ogawa.
\newblock Quantum hypothesis testing and the operational interpretation of the quantum {R}ényi relative entropies.
\newblock {\em Communications in Mathematical Physics}, 334(3):1617--1648, 2015.
\newblock \href {https://doi.org/10.1007/s00220-014-2248-x} {\path{doi:10.1007/s00220-014-2248-x}}.

\bibitem[MOS23]{mosek2023}
MOSEK ApS, Copenhagen, Denmark.
\newblock {\em The {MOSEK} Optimization Toolbox for {MATLAB} Manual}, 2023.
\newblock Available at \url{https://www.mosek.com}.

\bibitem[Nes83]{Nesterov1983}
Yurii Nesterov.
\newblock A method for solving the convex programming problem with convergence rate {${O}(1/k^2)$}.
\newblock {\em Soviet Mathematics Doklady}, 27(2):372--376, 1983.

\bibitem[Nes04]{Nesterov2004}
Yurii Nesterov.
\newblock {\em Introductory Lectures on Convex Optimization: A Basic Course}, volume~87 of {\em Applied Optimization}.
\newblock Springer, 2004.
\newblock \href {https://doi.org/10.1007/978-1-4419-8853-9} {\path{doi:10.1007/978-1-4419-8853-9}}.

\bibitem[NN94]{Nesterov1994}
Yurii Nesterov and Arkadi Nemirovskii.
\newblock {\em Interior-Point Polynomial Algorithms in Convex Programming}.
\newblock Studies in Applied Mathematics. SIAM, 1994.
\newblock \href {https://doi.org/10.1137/1.9781611970791} {\path{doi:10.1137/1.9781611970791}}.

\bibitem[Pet07]{petz2007quantum}
D{\'e}nes Petz.
\newblock {\em Quantum Information Theory and Quantum Statistics}.
\newblock Springer Science \& Business Media, 2007.
\newblock \href {https://doi.org/10.1007/978-3-540-74636-2} {\path{doi:10.1007/978-3-540-74636-2}}.

\bibitem[Pia09]{Piani2009}
Marco Piani.
\newblock Relative entropy of entanglement and restricted measurements.
\newblock {\em Physical Review Letters}, 103(16):160504, October 2009.
\newblock URL: \url{https://link.aps.org/doi/10.1103/PhysRevLett.103.160504}, \href {https://doi.org/10.1103/PhysRevLett.103.160504} {\path{doi:10.1103/PhysRevLett.103.160504}}.

\bibitem[Ser17]{Serafini17}
Alessio Serafini.
\newblock {\em Quantum Continuous Variables: A Primer of Theoretical Methods}.
\newblock CRC Press, 2017.
\newblock \href {https://doi.org/10.1201/9781315118727} {\path{doi:10.1201/9781315118727}}.

\bibitem[SGW26]{Sreekumar2025}
Sreejith Sreekumar, Ziv Goldfeld, and Mark~M. Wilde.
\newblock Performance guarantees for quantum neural estimation of entropies.
\newblock {\em {Quantum}}, 10:2113, May 2026.
\newblock \href {https://doi.org/10.22331/q-2026-05-21-2113} {\path{doi:10.22331/q-2026-05-21-2113}}.

\bibitem[Stu99]{sedumi1999}
Jos~F. Sturm.
\newblock Using {SeDuMi} 1.02, a {MATLAB} toolbox for optimization over symmetric cones.
\newblock {\em Optimization Methods and Software}, 11(1--4):625--653, 1999.

\bibitem[Tod01]{Todd2001}
Michael~J. Todd.
\newblock Semidefinite optimization.
\newblock {\em Acta Numerica}, 10:515--560, 2001.
\newblock \href {https://doi.org/10.1017/S096249290100005X} {\path{doi:10.1017/S096249290100005X}}.

\bibitem[VB96]{Vandenberghe1996}
Lieven Vandenberghe and Stephen Boyd.
\newblock Semidefinite programming.
\newblock {\em SIAM Review}, 38(1):49--95, 1996.
\newblock \href {https://doi.org/10.1137/1038003} {\path{doi:10.1137/1038003}}.

\bibitem[Wat18]{Watrous2018}
John Watrous.
\newblock {\em The Theory of Quantum Information}.
\newblock Cambridge University Press, 2018.
\newblock \href {https://doi.org/10.1017/9781316848142} {\path{doi:10.1017/9781316848142}}.

\bibitem[Wil25]{Wilde2025}
Mark~M. Wilde.
\newblock Quantum {F}isher information matrices from {R}\'enyi relative entropies, 2025.
\newblock URL: \url{https://arxiv.org/abs/2510.02218}, \href {https://arxiv.org/abs/2510.02218} {\path{arXiv:2510.02218}}.

\end{thebibliography}

\appendix

\section{Derivatives of scalar functions with respect to matrices}

\label{app:matrix-derivatives}In this appendix, we review the theory
of matrix derivatives and calculate matrix derivatives of various
functions that play a role in our paper.

Let $X$ be a $d\times d$ matrix. Let $f(X)$ be a scalar function
of a matrix. Then the matrix derivative (in numerator layout notation)
is defined as
\begin{equation}
\frac{\partial f(X)}{\partial X}\coloneqq\begin{bmatrix}\frac{\partial f(X)}{\partial x_{1,1}} & \frac{\partial f(X)}{\partial x_{2,1}} & \cdots & \frac{\partial f(X)}{\partial x_{d,1}}\\
\frac{\partial f(X)}{\partial x_{1,2}} & \frac{\partial f(X)}{\partial x_{2,2}} &  & \frac{\partial f(X)}{\partial x_{d,2}}\\
\vdots &  & \ddots & \vdots\\
\frac{\partial f(X)}{\partial x_{1,d}} & \frac{\partial f(X)}{\partial x_{2,d}} & \cdots & \frac{\partial f(X)}{\partial x_{d,d}}
\end{bmatrix},
\end{equation}
where the matrix elements of $X$ are denoted by $x_{i,j}$ for all
$i,j\in\left\{ 1,\ldots,d\right\} $.

\subsection{Power functions}
\begin{prop}
\label{prop:matrix-deriv-powers}Let $A$ and $X$ be $d\times d$
matrices, and let $n\in\mathbb{N}$. Then the following equality holds:
\begin{equation}
\frac{\partial\Tr[AX^{n}]}{\partial X}=\sum_{k=0}^{n-1}X^{k}AX^{n-k-1}.\label{eq:n-power-matrix-deriv}
\end{equation}
If $X$ is Hermitian with a spectral decomposition given by $X=\sum_{k}\lambda_{k}\Pi_{k}$,
then
\begin{equation}
\frac{\partial\Tr[AX^{n}]}{\partial X}=\sum_{\ell,m}f_{x^{n}}^{\left[1\right]}(\lambda_{\ell},\lambda_{m})\Pi_{\ell}A\Pi_{m},
\end{equation}
where $f_{x^{n}}^{\left[1\right]}$ is the first divided difference
of the function $x\mapsto x^{n}$ and is defined as 
\begin{equation}
f_{x^{n}}^{\left[1\right]}(\lambda_{\ell},\lambda_{m})\coloneqq\begin{cases}
n\lambda_{\ell}^{n-1} & :\lambda_{\ell}=\lambda_{m}\\
\frac{\lambda_{\ell}^{n}-\lambda_{m}^{n}}{\lambda_{\ell}-\lambda_{m}} & :\lambda_{\ell}\neq\lambda_{m}
\end{cases}.
\end{equation}
\end{prop}

\begin{proof}
By applying the product rule, consider that 
\begin{align}
\frac{\partial\Tr[AX^{n}]}{\partial x_{i,j}} & =\frac{\partial\Tr[AXX\cdots X]}{\partial x_{i,j}}\\
 & =\Tr\!\left[A\left(\frac{\partial}{\partial x_{i,j}}X\right)X\cdots X\right]+\Tr\!\left[AX\left(\frac{\partial}{\partial x_{i,j}}X\right)\cdots X\right]\nonumber \\
 & \qquad+\cdots+\Tr\!\left[AXX\cdots\left(\frac{\partial}{\partial x_{i,j}}X\right)\right]\\
 & =\Tr\!\left[X\cdots XA\left(\frac{\partial}{\partial x_{i,j}}X\right)\right]+\Tr\!\left[X\cdots XAX\left(\frac{\partial}{\partial x_{i,j}}X\right)\right]\nonumber \\
 & \qquad+\cdots+\Tr\!\left[AXX\cdots X\left(\frac{\partial}{\partial x_{i,j}}X\right)\right]\\
 & =\Tr\!\left[\left(\sum_{k=0}^{n-1}X^{k}AX^{n-k-1}\right)\left(\frac{\partial}{\partial x_{i,j}}X\right)\right]\\
 & =\Tr\!\left[\left(\sum_{k=0}^{n-1}X^{k}AX^{n-k-1}\right)|i\rangle\!\langle j|\right]\label{eq:simple-matrix-deriv}\\
 & =\langle j|\left(\sum_{k=0}^{n-1}X^{k}AX^{n-k-1}\right)|i\rangle\\
 & =\left[\sum_{k=0}^{n-1}X^{k}AX^{n-k-1}\right]_{j,i},
\end{align}
thus implying the claim in \eqref{eq:n-power-matrix-deriv}. The equality
in \eqref{eq:simple-matrix-deriv} follows because
\begin{align}
\frac{\partial}{\partial x_{i,j}}X & =\frac{\partial}{\partial x_{i,j}}\sum_{k,\ell}x_{k,\ell}|k\rangle\!\langle\ell|\\
 & =\sum_{k,\ell}\left(\frac{\partial}{\partial x_{i,j}}x_{k,\ell}\right)|k\rangle\!\langle\ell|\\
 & =\sum_{k,\ell}\delta_{i,k}\delta_{j,\ell}|k\rangle\!\langle\ell|\\
 & =|i\rangle\!\langle j|.
\end{align}

In the case that $X$ is Hermitian, consider that
\begin{align}
\frac{\partial\Tr[AX^{n}]}{\partial X} & =\sum_{k=0}^{n-1}X^{k}AX^{n-k-1}\\
 & =\sum_{k=0}^{n-1}\left(\sum_{\ell}\lambda_{\ell}\Pi_{\ell}\right)^{k}A\left(\sum_{m}\lambda_{m}\Pi_{m}\right)^{n-k-1}\\
 & =\sum_{k=0}^{n-1}\left(\sum_{\ell}\lambda_{\ell}^{k}\Pi_{\ell}\right)A\left(\sum_{m}\lambda_{m}^{n-k-1}\Pi_{m}\right)\\
 & =\sum_{\ell,m}\Pi_{\ell}A\Pi_{m}\sum_{k=0}^{n-1}\lambda_{\ell}^{k}\lambda_{m}^{n-k-1}\\
 & =\sum_{\ell,m}f_{x^{n}}^{\left[1\right]}(\lambda_{\ell},\lambda_{m})\Pi_{\ell}A\Pi_{m},
\end{align}
where the last equality follows because $\sum_{k=0}^{n-1}\lambda_{\ell}^{k}\lambda_{m}^{n-k-1}=f_{x^{n}}^{\left[1\right]}(\lambda_{\ell},\lambda_{m})$.
For a simple proof of this, see, e.g., \cite[Eqs.~(B12)--(B15)]{Wilde2025}.
\end{proof}
\begin{rem}
Special cases of Proposition~\ref{prop:matrix-deriv-powers} include
$n=1$ and $n=2$:
\begin{equation}
\frac{\partial\Tr[AX]}{\partial X}=A,\qquad\frac{\partial\Tr[AX^{2}]}{\partial X}=AX+XA.\label{eq:matrix-deriv-n-1-2}
\end{equation}
\end{rem}

\subsection{Functions with convergent power series expansions}

Suppose that $f$ is a function with a power series expansion convergent
on an interval $I$, and suppose furthermore that all the eigenvalues
of a Hermitian operator $X$ are contained in $I$. By reasoning similar
to that in \cite[Eqs.~(B18)--(B30)]{Wilde2025}, we then obtain the
following result by linearity:
\begin{equation}
\frac{\partial}{\partial X}\Tr[Af(X)]=\sum_{\ell,m}f^{[1]}(\lambda_{\ell},\lambda_{m})\Pi_{\ell}A\Pi_{m},\label{eq:matrix-deriv-func-f}
\end{equation}
where $X=\sum_{k}\lambda_{k}\Pi_{k}$ is a spectral decomposition
of $X$ and the first divided difference function $f^{[1]}(y_{1},y_{2})$
is defined for $y_{1},y_{2}\in I$ as
\begin{equation}
f^{[1]}(y_{1},y_{2})\coloneqq\begin{cases}
f'(y_{1}) & :y_{1}=y_{2}\\
\frac{f(y_{1})-f(y_{2})}{y_{1}-y_{2}} & :y_{1}\neq y_{2}
\end{cases}.
\end{equation}

By applying reasoning similar to that in the proofs of \cite[Propositions~47--50]{Wilde2025},
we obtain the following corollary:
\begin{cor}
For all Hermitian $X$,
\begin{equation}
\frac{\partial}{\partial X}\Tr[Ae^{X}]=\int_{0}^{1}dt\ e^{tX}Ae^{\left(1-t\right)X}.
\end{equation}
For all positive definite $X$,
\begin{align}
\frac{\partial}{\partial X}\Tr[A\ln X] & =\int_{0}^{\infty}ds\ \left(X+sI\right)^{-1}A\left(X+sI\right)^{-1},\label{eq:deriv-log}\\
\forall r\in\mathbb{R},\qquad\frac{\partial}{\partial X}\Tr[AX^{r}] & =r\int_{0}^{1}dt\int_{0}^{\infty}ds\ \left(\frac{X^{rt}}{X+sI}\right)A\left(\frac{X^{r\left(1-t\right)}}{X+sI}\right),\\
\forall r\in\left(-1,0\right)\cup\left(0,1\right)\qquad\frac{\partial}{\partial X}\Tr[AX^{r}] & =\frac{\sin(r\pi)}{\pi}\int_{0}^{\infty}dt\ t^{r}\left(X+tI\right)^{-1}A\left(X+tI\right)^{-1}.\label{eq:power-func-deriv-minus-1-to-plus-1}
\end{align}
\end{cor}

\section{Derivations of integral formulas}

\label{app:Derivations-of-integral-formulas}In this appendix, we
derive the integral formulas in \eqref{eq:key-integral-meas-renyi-eigenval-bnds-first}
and \eqref{eq:key-integral-meas-renyi-eigenval-bnds}.

\subsection{Proof of \eqref{eq:key-integral-meas-renyi-eigenval-bnds-first}}

Set $\zeta\coloneqq\frac{\alpha}{\alpha-1}$, so that $\zeta\in\left(-1,0\right)$
for $\alpha\in\left(0,\frac{1}{2}\right)$ and
\[
-c_{1}(\alpha)\int_{0}^{\infty}dt\ t^{\frac{\alpha}{\alpha-1}}\left(k+t\right)^{-3}=-\left(1-\alpha\right)\frac{\sin\!\left(\zeta\pi\right)}{\pi}\int_{0}^{\infty}dt\ t^{\zeta}\left(k+t\right)^{-3}.
\]
Consider that
\begin{equation}
\left(k+t\right)^{-3}=\frac{1}{2}\frac{\partial}{\partial k^{2}}\left(k+t\right)^{-1},\label{eq:derivative-trick-cube}
\end{equation}
so that
\begin{align}
\int_{0}^{\infty}dt\ t^{\zeta}\left(k+t\right)^{-3} & =\int_{0}^{\infty}dt\ t^{\zeta}\left(\frac{1}{2}\frac{\partial}{\partial k^{2}}\left(k+t\right)^{-1}\right)\\
 & =\frac{1}{2}\frac{\partial}{\partial k^{2}}\int_{0}^{\infty}dt\ \frac{t^{\zeta}}{k+t}\\
 & =\frac{1}{2}\frac{\partial}{\partial k^{2}}\left(-k^{\zeta}\frac{\pi}{\sin\!\left(\zeta\pi\right)}\right),
\end{align}
where the last equality follows as a consequence of the following
standard integral (see, e.g., \cite[Eq.~(B.117)]{Wilde2025}), holding
for $r\in\left(-1,0\right)$ and $x>0$:
\begin{align}
x^{r} & =-\frac{\sin(r\pi)}{\pi}\int_{0}^{\infty}dt\ \frac{t^{r}}{x+t}.\label{eq:standard-integral-power-sine}
\end{align}
So it follows that
\begin{align}
-\left(1-\alpha\right)\frac{\sin\!\left(\zeta\pi\right)}{\pi}\int_{0}^{\infty}dt\ t^{\zeta}\left(k+t\right)^{-3} & =-\left(1-\alpha\right)\frac{\sin\!\left(\zeta\pi\right)}{2\pi}\frac{\partial}{\partial k^{2}}\left(-k^{\zeta}\frac{\pi}{\sin\!\left(\zeta\pi\right)}\right)\\
 & =\frac{\left(1-\alpha\right)}{2}\frac{\partial}{\partial k^{2}}k^{\zeta}\\
 & =\frac{\left(1-\alpha\right)\zeta\left(\zeta-1\right)}{2}k^{\zeta-2}\\
 & =\frac{\left(1-\alpha\right)\frac{\alpha}{\alpha-1}\left(\frac{\alpha}{\alpha-1}-1\right)}{2}k^{\frac{\alpha}{\alpha-1}-2}\\
 & =\frac{\alpha}{2\left(1-\alpha\right)}k^{-\frac{2-\alpha}{1-\alpha}},
\end{align}
which completes the proof of \eqref{eq:key-integral-meas-renyi-eigenval-bnds-first}.

\subsection{Proof of \eqref{eq:key-integral-meas-renyi-eigenval-bnds}}

Set $\xi\coloneqq\frac{\alpha-1}{\alpha}$, so that $\xi\in\left[-1,0\right)$
for $\alpha\in\left[\frac{1}{2},1\right)$. Consider that
\begin{equation}
-c_{2}(\alpha)\int_{0}^{\infty}dt\ t^{\frac{\alpha-1}{\alpha}}\left(k+t\right)^{-3}=-\alpha\frac{\sin\!\left(\xi\pi\right)}{\pi}\int_{0}^{\infty}dt\ t^{\xi}\left(k+t\right)^{-3}.
\end{equation}
Apply \eqref{eq:derivative-trick-cube} again to conclude that
\begin{align}
-\alpha\frac{\sin\!\left(\xi\pi\right)}{\pi}\int_{0}^{\infty}dt\ t^{\xi}\left(k+t\right)^{-3} & =-\alpha\frac{\sin\!\left(\xi\pi\right)}{\pi}\frac{1}{2}\frac{\partial}{\partial k^{2}}\int_{0}^{\infty}dt\ \frac{t^{\xi}}{k+t}\\
 & =-\alpha\frac{\sin\!\left(\xi\pi\right)}{\pi}\frac{1}{2}\frac{\partial}{\partial k^{2}}\left(-k^{\xi}\frac{\pi}{\sin\!\left(\xi\pi\right)}\right)\\
 & =\frac{\alpha}{2}\frac{\partial}{\partial k^{2}}\left(k^{\xi}\right)\\
 & =\frac{\alpha\xi\left(\xi-1\right)}{2}k^{\xi-2}\\
 & =\frac{\alpha\frac{\alpha-1}{\alpha}\left(\frac{\alpha-1}{\alpha}-1\right)}{2}k^{\frac{\alpha-1}{\alpha}-2}\\
 & =\left(\frac{1-\alpha}{2\alpha}\right)k^{-\frac{\alpha+1}{\alpha}},
\end{align}
where the second equality follows again from \eqref{eq:standard-integral-power-sine}.
This proves \eqref{eq:key-integral-meas-renyi-eigenval-bnds} for
$\alpha\in\left[\frac{1}{2},1\right)$. 

For $\alpha>1$, consider that $\xi\in\left(0,1\right)$, and then
\begin{align}
-c_{2}(\alpha)\int_{0}^{\infty}dt\ t^{\frac{\alpha-1}{\alpha}}\left(k+t\right)^{-3} & =-\alpha\frac{\sin\!\left(\xi\pi\right)}{\pi}\int_{0}^{\infty}dt\ t^{\xi}\left(k+t\right)^{-3}\\
 & =-\alpha\frac{\sin\!\left(\xi\pi\right)}{\pi}\int_{0}^{\infty}dt\ t^{\xi}\left(-\frac{1}{2}\right)\frac{\partial}{\partial k}\left(k+t\right)^{-2}\\
 & =\alpha\frac{\sin\!\left(\xi\pi\right)}{2\pi}\frac{\partial}{\partial k}\int_{0}^{\infty}dt\ \frac{t^{\xi}}{\left(k+t\right)^{2}}\\
 & =\alpha\frac{\sin\!\left(\xi\pi\right)}{2\pi}\frac{\partial}{\partial k}\left(\frac{\xi\pi}{\sin(\xi\pi)}k^{\xi-1}\right)\\
 & =\alpha\frac{\xi}{2}\frac{\partial}{\partial k}\left(k^{\xi-1}\right)\\
 & =\alpha\frac{\xi\left(\xi-1\right)}{2}k^{\xi-2}\\
 & =\left(\frac{1-\alpha}{2\alpha}\right)k^{-\frac{\alpha+1}{\alpha}}
\end{align}
where we invoked the following integral representation for $r\in(0,1)$
and $x,y>0$ such that $x\neq y$ (see, e.g., \cite[Eq.~(B.102)]{Wilde2025}):
\begin{align}
\frac{x^{r}-y^{r}}{x-y} & =\frac{\sin(r\pi)}{\pi}\int_{0}^{\infty}dt\ \frac{t^{r}}{\left(x+t\right)\left(y+t\right)},\label{eq:integral-rep-div-diff-x-r}
\end{align}
which implies in the $y\to x$ limit that
\begin{equation}
rx^{r-1}=\frac{\sin(r\pi)}{\pi}\int_{0}^{\infty}dt\ \frac{t^{r}}{\left(x+t\right)^{2}}.
\end{equation}
This concludes the proof of \eqref{eq:key-integral-meas-renyi-eigenval-bnds}
for $\alpha>1$.
\end{document}